\documentclass[journal]{IEEEtran}

\usepackage{amsmath,amssymb,amsfonts}
\usepackage{array}
\usepackage[caption=false,font=normalsize,labelfont=sf,textfont=sf]{subfig}
\usepackage{textcomp}
\usepackage{stfloats}
\usepackage{verbatim}
\usepackage{graphicx}
\usepackage{cite}
\usepackage{tabularx}
\usepackage{booktabs}

\def\BibTeX{{\rm B\kern-.05em{\sc i\kern-.025em b}\kern-.08em
    T\kern-.1667em\lower.7ex\hbox{E}\kern-.125emX}}

\usepackage{acronym}
\acrodef{5G}{the fifth generation}
\acrodef{MIMO}{multiple-input multiple-output}
\acrodef{MISO}{multiple-input single-output}
\acrodef{MU}{multi-user}
\acrodef{EE}{energy efficiency}
\acrodef{SE}{spectral efficiency}
\acrodef{RF}{radio-frequency}
\acrodef{BS}{base station}
\acrodef{UE}{user equipment}
\acrodef{AQNM}{additive quantization noise model}
\acrodef{DAC}{digital-to-analog converter}
\acrodef{PA}{power amplifier}
\acrodef{VGA}{variable-gain amplifier}
\acrodef{SINR}{signal-to-interference-plus-noise ratio}
\acrodef{WMMSE}{weighted minimum mean square error}
\acrodef{LoS}{line-of-sight}
\acrodef{NLoS}{non-line-of-sight}
\acrodef{AoA}{angle-of-arrival}
\acrodef{AoD}{angle-of-departure}
\acrodef{UPA}{uniform planar array}
\acrodef{ARV}{array response vector}
\acrodef{CGV}{channel gain vector}
\acrodef{AGV}{antenna gain vector}
\acrodef{EM}{electromagnetic}
\acrodef{CSI}{channel state information}
\acrodef{OFDM}{orthogonal frequency-division multiplexing}
\acrodef{RIS}{reconfigurable intelligent surface}
\acrodef{MA}{movable antenna}
\acrodef{3D}{three-dimensional}
\acrodef{AWGN}{additive white Gaussian noise}
\acrodef{ERA}{electromagnetically reconfigurable antenna}
\acrodef{MiLAC}[MiLAC]{microwave linear analog computer}
\acrodef{LMI}[LMI]{linear matrix inequality}
\acrodef{SDP}[SDP]{semi-definite programming}
\acrodef{PGD}[PGD]{projected gradient descent}
\acrodef{SVD}[SVD]{singular value decomposition}
\acrodef{MOO}[MOO]{multi-objective optimization}
\acrodef{KKT}[KKT]{Karush-Kuhn-Tucker}
\acrodef{FP}[FP]{fractional programming}
\acrodef{PS}[PS]{phase shifter}
\acrodef{SC}[SC]{sub-connected}
\acrodef{FC}[FC]{fully-connected}
\acrodef{SCA}[SCA]{successive convex approximation}
\acrodef{SOC}[SOC]{second-order cone}
\usepackage{amsthm}
\newtheoremstyle{zackplain}%
  {\topsep}{\topsep}{\itshape}{}{\bfseries}{:}{.5em}{}
\theoremstyle{zackplain}
\newtheorem{lemma}{\textbf{Lemma}}

\newtheorem{proposition}{\textbf{Proposition}}

\newtheorem{remark}{\textbf{Remark}}

\renewenvironment{proof}{\textit{\textbf{Proof:}}}{\hfill$\square$}
\usepackage{color}

\usepackage{tikz}
\usetikzlibrary{calc, positioning, fit, backgrounds, decorations.pathmorphing, arrows.meta}
\tikzset{
  antenna/.pic={
    \draw[thick] (0,0) -- (0.15,0) -- (0.15,0.32);
    \fill[black] (0.05,0.46) -- (0.25,0.46) -- (0.15,0.32) -- cycle;
  }
}
\usepackage[utf8]{inputenc}
\usepackage{pgfplots}
\pgfplotsset{compat=1.18}
\usepackage{tabu,longtable}
\usepackage{diagbox}
\usepackage{threeparttable}
\usepackage{makecell}
\usepackage{multirow} 
\usepackage[bottom=1.02in,top=0.75in,left=0.625in,right=0.625in]{geometry}
\usepackage{units} 		
\usepackage{bm} 		
\usepackage{amssymb} 	
\newcommand{\TT}{\mathsf{T}}
\newcommand{\HH}{\mathsf{H}}

\newcommand{\av}{{\bf a}}

\newcommand{\cv}{{\bf c}}

\newcommand{\fv}{{\bf f}}

\newcommand{\hv}{{\bf h}}

\newcommand{\nv}{{\bf n}}

\newcommand{\pv}{{\bf p}}
\newcommand{\qv}{{\bf q}}

\newcommand{\sv}{{\bf s}}

\newcommand{\uv}{{\bf u}}
\newcommand{\wv}{{\bf w}}

\newcommand{\xv}{{\bf x}}

\newcommand{\zv}{{\bf z}}
\newcommand{\zerov}{{\bf 0}}
\newcommand{\onev}{{\bf 1}}

\newcommand{\Am}{{\bf A}}
\newcommand{\Bm}{{\bf B}}
\newcommand{\Cm}{{\bf C}}
\newcommand{\Dm}{{\bf D}}

\newcommand{\Fm}{{\bf F}}
\newcommand{\Gm}{{\bf G}}
\newcommand{\Hm}{{\bf H}}
\newcommand{\Id}{{\bf I}}

\newcommand{\Lm}{{\bf L}}

\newcommand{\Qm}{{\bf Q}}
\newcommand{\Rm}{{\bf R}}

\newcommand{\Um}{{\bf U}}
\newcommand{\Wm}{{\bf W}}
\newcommand{\Vm}{{\bf V}}
\newcommand{\Xm}{{\bf X}}
\newcommand{\Ym}{{\bf Y}}
\newcommand{\Zm}{{\bf Z}}

\newcommand{\omegav}{\hbox{\boldmath$\omega$}}

\newcommand{\Pim}{\hbox{\boldmath$\Pi$}}

\newcommand{\Thetam}{\hbox{\boldmath$\Theta$}}

\usepackage{amsmath}

\usepackage{algorithm}
\usepackage{algpseudocode}
\algtext*{EndWhile}
\algtext*{EndIf}
\algtext*{EndFor}

\makeatletter
\algnewcommand{\LineComment}[1]{\Statex \hskip\ALG@thistlm \(\triangleright\) #1}
\makeatother
\usepackage{xcolor}
\usepackage{tcolorbox}

\definecolor{RoyalBlue}{RGB}{65, 105, 225}  

\begin{document}

\title{Quantization-Aware EE Optimization and SE-EE Tradeoff for MiLAC-Aided MU-MISO Beamforming}

\author{Yuchen~Zhang, \emph{Member, IEEE}, Pinjun~Zheng, \emph{Member, IEEE}, and Tareq~Y.~Al-Naffouri, \emph{Fellow, IEEE}
\thanks{
This publication is based upon work supported by King Abdullah University of Science and Technology (KAUST) under Award No. ORFS-CRG12-2024-6478 and Global Fellowship Program under Award No. RFS-2025-6844.

Y. Zhang and T. Y. Al-Naffouri are with the Computer, Electrical and Mathematical Sciences and Engineering Division, King Abdullah University of Science and Technology (KAUST), Thuwal 23955-6900, Saudi Arabia (e-mail: \{yuchen.zhang, tareq.alnaffouri\}@kaust.edu.sa).}%
\thanks{P. Zheng is with the School of Engineering, The University of British Columbia, Kelowna, BC V1V 1V7, Canada (e-mail: pinjun.zheng@ubc.ca).}
}

\maketitle

\begin{abstract}
In large antenna arrays, hardware power consumption becomes a dominant design constraint, making energy efficiency (EE) a primary design objective alongside spectral efficiency (SE). Microwave linear analog computer (MiLAC)-aided beamforming, whose front end is a passive reciprocal stream-to-antenna network, emerges as a promising candidate to address this tradeoff by reducing the number of active radio-frequency chains to the stream number, at a moderate SE cost. Despite this promise, no EE optimization framework has been established for MiLAC-aided beamforming that accounts for digital-to-analog converter quantization noise and post-quantized transmit power. We fill this gap for downlink multiuser multiple-input single-output systems by formulating quantization-aware EE maximization over the MiLAC-feasible beamformer and characterizing the resulting SE-EE tradeoff. Three contributions follow. First, we prove a row-space optimality property of the effective MiLAC-aided beamformer, yielding an equivalent reduced-dimension reformulation whose complexity scales with the stream number rather than the antenna number. Second, we develop a low-complexity Dinkelbach-weighted minimum mean-square error algorithm aided by projected gradient descent that is guaranteed to converge to a stationary point. Third, we cast the SE-EE tradeoff as a multi-objective optimization problem and trace the tradeoff boundary via the weighted-sum method. An alternative reduced-dimension variable removes the bilinearity in the original parameterization, and an auxiliary-variable lift combined with successive convex approximation yields a convex per-iteration subproblem with guaranteed convergence. Numerical results on a DeepMIMO v4 deployment demonstrate that MiLAC-aided beamforming substantially improves EE over digital and hybrid benchmarks at a moderate SE cost. The tradeoff boundary analysis further shows that MiLAC-aided beamforming significantly expands the achievable SE-EE operating region relative to digital and hybrid architectures.
\end{abstract}

\begin{IEEEkeywords}
Beamforming, DAC quantization, energy efficiency, MiLAC, MU-MISO, PGD, SE-EE tradeoff, WMMSE.
\end{IEEEkeywords}

\looseness=-1 \section{Introduction}
Gigantic \ac{MIMO} has emerged as a promising direction for sixth-generation (6G) wireless networks~\cite{Wang2023Road6G,Bjornson2025Gigantic,Ning2026PMI}.
Scaling the antenna number from tens to hundreds or thousands proportionally inflates the transmitter-side power consumed by \ac{RF} chains, \acp{DAC}, and analog front-end circuitry~\cite{Larsson2014Massive,Bjornson2017Massive,Choi2022Energy,Ribeiro2018Energy}.
In digital architectures, both the chain number and the associated circuit power grow linearly with the antenna number $N$, so the power overhead can exceed the transmit power itself at large $N$.
Hybrid analog-digital architectures reduce the active chain number to $N_{\mathrm{RF}} < N$ but introduce a \ac{PS} network whose power depends on the connectivity pattern: \ac{FC} hybrid beamforming requires $N N_{\mathrm{RF}}$ \acp{PS}, while \ac{SC} hybrid beamforming reduces the number to $N$~\cite{Sohrabi2016Hybrid,Ribeiro2018Energy}.
Because these architectures differ substantially in how their hardware power scales with $N$, beamforming design must be judged not only by the \ac{SE} it delivers but also by the \ac{EE}, the \ac{SE} obtained per Joule.

A \ac{MiLAC} is a tunable multiport microwave network that processes signals linearly as they propagate through it~\cite{Nerini2025AnalogI,Nerini2025AnalogII}.
For beamforming, $K$ active \ac{RF} chains, each fed by its \ac{DAC} pair, drive the stream-side ports while $N$ antennas connect to the remaining ports, so the stream-to-antenna mapping is executed at electromagnetic speed through a passive network rather than in digital or active analog hardware~\cite{Nerini2025AnalogII}.
Practical considerations restrict the admittances to purely imaginary (lossless) and symmetric (reciprocal) values, since nonzero real parts would either dissipate power or require active components, and asymmetric admittances would require non-reciprocal devices~\cite{Wu2026MiLAC}.
Under this restriction, the scattering matrix is unitary and symmetric, constraining the feasible beamformer set~\cite{Wu2026MiLAC}. This paper adopts the lossless reciprocal \ac{MiLAC} throughout, where it serves only to realize this linear stream-to-antenna mapping. The general-purpose analog computations studied in~\cite{Nerini2025AnalogI}, such as the linear minimum mean square error estimator and matrix inversion, are not involved.

Two features distinguish \ac{MiLAC}-aided beamforming from the digital and hybrid architectures above.
First, its stream-to-antenna mapping is computed by a passive and purely reactive network. The lossless admittances dissipate no power, so the network consumes power only in the drive circuits that hold each admittance at its assigned value, and contains no per-element active circuitry, unlike the \ac{PS} networks of hybrid beamforming, whose phase shifters are active components that draw supply power in the signal path~\cite{Choi2022Energy,Ribeiro2018Energy}. The same separation between lossless elements and their drive circuits has been measured on \ac{RIS} prototypes~\cite{Wang2024RISPower}.
Second, its active front end consists of $K$ \ac{RF} chains, one per stream, whose \acp{DAC} carry the stream symbols directly. This is fewer than the $N$ chains of digital beamforming and comparable to the $N_{\mathrm{RF}}$ chains of hybrid beamforming, since $N_{\mathrm{RF}}$ is usually close to $K$. What distinguishes the two architectures is therefore not the chain number but what each obtains from it. For the same chain number, \ac{MiLAC}-aided beamforming attains comparable or superior rate performance, requires no symbol-level digital beamforming because the stream-to-antenna mapping is performed by the network, and is particularly suited to low-resolution \acp{DAC}~\cite{Wu2026MiLAC}, while replacing the powered \ac{PS} network by a passive one.
The drive power of the passive network nevertheless grows with the network size (quadratically in $N{+}K$), so the architectural advantage is a reallocation of $N$-dependent power from active to passive hardware rather than a complete decoupling from the antenna number.
Whether this reallocation yields an \ac{EE} gain, and over what operating regime, is the quantitative question this paper answers.

\looseness=-1 The \ac{MiLAC} concept has developed rapidly in two directions.
The first treats \ac{MiLAC} as a general-purpose analog computer: the foundational theory in~\cite{Nerini2025AnalogI,Nerini2025AnalogII} established that, with unconstrained admittances, a tunable multiport network can implement general linear operations with favorable computational complexity.
The second applies \ac{MiLAC} to beamforming under lossless and reciprocal constraints.
In the point-to-point setting, lossless reciprocal MiLACs achieve full \ac{MIMO} capacity~\cite{Nerini2025Capacity}, and reduced-complexity stem-connected architectures lower the admittance number from $\mathcal{O}((N{+}K)^2)$ to $\mathcal{O}(NK)$ while preserving capacity~\cite{Nerini2025Reduced}.
In the multiuser multiple-input single-output (MU-MISO) setting, the feasible beamformer set was characterized and beamforming designs were proposed~\cite{Wu2026MiLAC}, while a performance-limit analysis showed that lossless reciprocal MiLACs cannot replicate every fully digital beamformer, although the \ac{SE} gap shrinks as $N/K$ grows~\cite{Fang2026Performance}.

Complementary efforts address analog-domain channel estimation~\cite{Zhang2026Channel,Zhang2026ChannelBF}, physics-compliant mutual-coupling-aware modeling~\cite{Nerini2026Physics}, and hardware realizations using hybrid couplers and \acp{PS}~\cite{Nerini2026HybridCouplers}. Beyond \ac{MiLAC}, stacked intelligent metasurfaces likewise process signals in the wave domain, using multiple stacked programmable metasurface layers for holographic \ac{MIMO} communications~\cite{An2023SIM}.

Despite this progress, all existing \ac{MiLAC} studies optimize for \ac{SE}. None addresses the power-consumption accounting essential for \ac{EE} analysis.
Three elements are missing: (i)~a transmitter-side power model that captures \ac{DAC} quantization noise via the \ac{AQNM} and accounts for post-quantized transmit power, (ii)~an \ac{EE}-oriented solution that exploits \ac{MiLAC}-aided beamforming's unique stream-domain signal structure, and (iii)~a systematic characterization of the operating regime in which \ac{MiLAC}-aided beamforming's hardware advantage translates into a tangible \ac{SE}-\ac{EE} benefit relative to digital and hybrid benchmarks. The \ac{MiLAC} network and the \acp{DAC} must be treated jointly because the architecture itself ties them together. A \ac{MiLAC}-aided transmitter contains no digital beamforming: each of the $K$ streams passes through its own \ac{DAC} pair, is amplified, and is fed into the passive network, which computes the stream-to-antenna mapping in the analog domain. The quantization distortion therefore enters at the streams, before the mapping, and its covariance is correlated across antennas instead of diagonal as under per-antenna quantization. For the same reason, only $K$ \acp{DAC}, instead of $N$, contribute to the power in the \ac{EE} denominator. Both consequences follow from where the architecture places the \acp{DAC} relative to the network, and both shape the formulation and the solution developed below for \ac{MiLAC}-aided beamforming.

\looseness=-1 \ac{EE} optimization for large-array architectures has been studied extensively.
Foundational fractional-programming frameworks appear in~\cite{Zappone2015Energy,Ng2012Energy}, and scaling laws for very large arrays are analyzed in~\cite{Ngo2013Energy}.
For fully digital \ac{MIMO}, jointly optimal antenna/user/power \ac{EE} design is studied in~\cite{Bjornson2015Optimal}, joint \ac{EE}--\ac{SE} analysis is treated in~\cite{Bjornson2017Massive}, coordinated multi-cell \ac{EE} precoding is developed in~\cite{He2014Energy}, and the interplay between \ac{DAC} resolution and \ac{EE} is analyzed in~\cite{Choi2022Energy,Ribeiro2018Energy,Orhan2015Energy}.
For hybrid beamforming, low-complexity designs that reduce front-end power appear in~\cite{Sohrabi2016Hybrid,Liang2014LowComplexity}, and architecture-level \ac{EE} comparisons accounting for \ac{PS} power are reported in~\cite{MendezRial2016Hybrid,Ribeiro2018Energy}.
For \ac{RIS}-aided systems, which share the passive multiport hardware philosophy, \ac{EE} and power-consumption frameworks appear in~\cite{Zhou2024Optimizing,Wang2024RISPower}, and \ac{EE} maximization for dynamic metasurface antennas, which realize large arrays with reduced hardware power, is studied in~\cite{You2023DMA}.
None of these works addresses \ac{MiLAC} under the combination of transmitter-side \ac{DAC}-\ac{AQNM}, post-quantized transmit power accounting, and a static-power model for the drive circuits of the passive network.

This paper aims to answer two coupled questions: \emph{how should \ac{MiLAC}-aided beamforming be optimized for \ac{EE} under quantization-aware modeling, and in which operating regimes does that optimization expose a meaningful \ac{SE}-\ac{EE} tradeoff relative to digital and hybrid benchmarks?}
The main contributions are as follows.

\begin{itemize}
    \item \textbf{Quantization-aware \ac{EE} formulation for \ac{MiLAC}:}
    We introduce the first \ac{EE} maximization formulation for MU-MISO lossless reciprocal \ac{MiLAC}-aided beamforming, coupling the transmitter-side \ac{DAC}-\ac{AQNM} signal model with a post-quantized transmit-power constraint and a power consumption model that captures the \ac{PA} supply, stream-scaled \ac{RF}/\ac{DAC} circuit power, and \ac{MiLAC} circuit-incurred static power.\footnote{\looseness=-1 Online baseband-compute power is omitted because it is either common to all architectures or too implementation-specific to sharpen the comparison.}

    \item \textbf{Row-space reduction and low-complexity \ac{EE} solution:}
    We prove row-space optimality and derive an equivalent reduced-dimension parameterization that preserves both the achievable rate and the post-quantized power.
    This reduction exploits a structural property unique to \ac{MiLAC}: stream-domain \ac{DAC} placement yields a full-covariance \ac{AQNM} distortion compatible with row-space projection, a property that digital/hybrid beamforming lacks.
    On the reduced-dimension problem, we develop a low-complexity Dinkelbach-\ac{WMMSE} algorithm with closed-form weighted power updates and a \ac{PGD}-based $\Ym$ update, reducing the dominant per-iteration cost from $N$- to $K$-scaling.

    \item \textbf{\ac{SE}-\ac{EE} tradeoff characterization:}
    We formulate joint \ac{SE} and \ac{EE} maximization as a \ac{MOO} problem and adopt the weighted-sum method to trace the tradeoff boundary. Each weighted-sum subproblem is recast in an equivalent $K\times K$ reduced variable that eliminates the bilinearity, then lifted with auxiliary variables and solved via \ac{SCA}. The per-iteration subproblem is convex and every limit point of the iterates is proved to satisfy the \ac{KKT} conditions.

    \item \textbf{Site-specific evaluation and operating-regime insights:}
    Extensive simulations on a DeepMIMO v4 deployment show that \ac{MiLAC} attains significantly higher \ac{EE} than digital and both \ac{FC} and \ac{SC} hybrid benchmarks across a transmit-power sweep, at a moderate \ac{SE} cost.
    The tradeoff boundary comparison further reveals that \ac{MiLAC} achieves a substantially expanded \ac{SE}-\ac{EE} operating region that digital and hybrid architectures cannot match.
\end{itemize}

The rest of the paper is organized as follows.
Section~\ref{sec:system_model} introduces the \ac{MiLAC} system model and the quantization-aware \ac{EE} formulation.
Section~\ref{sec:reformulation} formulates the optimization problem and identifies the reduced-dimension property.
Section~\ref{sec:lc_solver} develops the low-complexity Dinkelbach-\ac{WMMSE} solution.
Section~\ref{sec:continuation} formulates the \ac{SE}-\ac{EE} tradeoff as a \ac{MOO} problem and develops the weighted-sum solver.
Section~\ref{sec:results} reports benchmark comparisons.
Section~\ref{sec:conclusion} concludes.

\looseness=-1 \emph{Notations:}
Scalars are denoted by lowercase letters, vectors by bold lowercase letters, and matrices by bold uppercase letters.
The Euclidean norm and Frobenius norm are $\|\cdot\|$ and $\|\cdot\|_{\mathrm{F}}$, and $\|\cdot\|_2$ is the spectral norm.
Transpose and Hermitian transpose are $(\cdot)^{\TT}$ and $(\cdot)^{\HH}$.
For a matrix $\Am$, $\operatorname{tr}(\Am)$ is the trace, $\operatorname{diag}(\av)$ is the diagonal matrix with entries $\av$, and $\Am \preceq \Bm$ means $\Bm - \Am$ is positive semidefinite.
$\Re\{a\}$ and $\bar{a}$ denote the real part and complex conjugate, respectively.
$\odot$ is the Hadamard product and $|\Am|^{\odot 2}$ the matrix of elementwise squared magnitudes.
$\operatorname{Ran}(\Am)$ and $\operatorname{Null}(\Am)$ are the range and null spaces of $\Am$.
$\Id_n$ is the $n \times n$ identity matrix, $[x]_+ = \max(x,0)$, and $\mathcal{CN}(\av,\Cm)$ denotes a circularly symmetric complex Gaussian distribution with mean $\av$ and covariance $\Cm$.

\begin{figure*}[t]
\centering
\begin{tikzpicture}[
  >=Stealth, font=\small,
  every node/.style={inner sep=1pt},
  block/.style={draw, rounded corners=3pt, minimum height=1.05cm, minimum width=1.7cm,
                align=center, line width=0.8pt, fill=blue!6},
  port/.style={circle, draw, fill=white, line width=0.5pt, minimum size=4.5pt, inner sep=0pt},
  annot/.style={font=\footnotesize, align=center},
  smallannot/.style={font=\scriptsize, align=center},
  flow/.style={->, line width=0.9pt},
  zonelabel/.style={font=\scriptsize\bfseries},
  chainblk/.style={minimum height=0.62cm, font=\small, align=center}
]

\def\rA{1.20}\def\rB{0.40}\def\rC{-1.20}\def\rD{-0.40}
\def\aA{1.20}\def\aB{0.62}\def\aC{0.04}\def\aD{-1.20}\def\aE{-0.56}
\def\xs{0.45}\def\xdac{2.90}\def\xpa{5.60}\def\xant{14.15}\def\xusr{16.85}

\node[block, chainblk, minimum width=1.76cm] (dacA) at (\xdac,\rA) {$b$-bit DAC};
\node[block, chainblk, minimum width=1.76cm] (dacB) at (\xdac,\rB) {$b$-bit DAC};
\node[block, chainblk, minimum width=1.76cm] (dacC) at (\xdac,\rC) {$b$-bit DAC};
\node[block, chainblk, minimum width=1.00cm, fill=blue!16] (paA) at (\xpa,\rA) {PA};
\node[block, chainblk, minimum width=1.00cm, fill=blue!16] (paB) at (\xpa,\rB) {PA};
\node[block, chainblk, minimum width=1.00cm, fill=blue!16] (paC) at (\xpa,\rC) {PA};
\node[font=\footnotesize] at (\xdac,\rD) {$\vdots$};
\node[font=\footnotesize] at (\xpa ,\rD) {$\vdots$};
\coordinate (fecenter) at (4.25,0);

\node[draw, rounded corners=3pt, fill=orange!8, line width=1.0pt,
      minimum width=4.6cm, minimum height=3.4cm] (milac) at (9.70,0) {};
\node[font=\small\bfseries] at ([yshift=-4.2mm]milac.north) {$(N{+}K)$-port MiLAC};

\coordinate (mc) at ([yshift=-1mm]milac.center);
\coordinate (p1) at ($(mc)+(-0.95, 0.55)$);
\coordinate (p2) at ($(mc)+( 0.95, 0.55)$);
\coordinate (p3) at ($(mc)+( 0.95,-0.60)$);
\coordinate (p4) at ($(mc)+(-0.95,-0.60)$);
\node[port] (n1) at (p1) {};
\node[port] (n2) at (p2) {};
\node[port] (n3) at (p3) {};
\node[port] (n4) at (p4) {};

\node[font=\scriptsize] at ([xshift=-2mm,yshift= 2mm]n1) {1};
\node[font=\scriptsize] at ([xshift= 2mm,yshift= 2mm]n2) {2};
\node[font=\scriptsize] at ([xshift= 2mm,yshift=-2mm]n3) {3};
\node[font=\scriptsize] at ([xshift=-2mm,yshift=-2mm]n4) {4};

\draw[line width=0.5pt] (n1) -- (n2);
\node[draw, line width=0.5pt, fill=white, inner sep=0pt,
      minimum width=3mm, minimum height=1.8mm]
  (adm12) at ($(n1)!0.5!(n2)$) {};
\node[font=\tiny, inner sep=1pt, above=0.2pt of adm12] {$\tilde{Y}_{12}$};

\draw[line width=0.5pt] (n4) -- (n3);
\node[draw, line width=0.5pt, fill=white, inner sep=0pt,
      minimum width=3mm, minimum height=1.8mm]
  (adm43) at ($(n4)!0.5!(n3)$) {};
\node[font=\tiny, inner sep=1pt, below=0.2pt of adm43] {$\tilde{Y}_{43}$};

\draw[line width=0.5pt] (n1) -- (n4);
\node[draw, line width=0.5pt, fill=white, inner sep=0pt,
      minimum width=1.8mm, minimum height=3mm]
  (adm14) at ($(n1)!0.5!(n4)$) {};
\node[font=\tiny, inner sep=1pt, left=0.8pt of adm14] {$\tilde{Y}_{14}$};

\draw[line width=0.5pt] (n2) -- (n3);
\node[draw, line width=0.5pt, fill=white, inner sep=0pt,
      minimum width=1.8mm, minimum height=3mm]
  (adm23) at ($(n2)!0.5!(n3)$) {};
\node[font=\tiny, inner sep=1pt, right=0.8pt of adm23] {$\tilde{Y}_{23}$};

\draw[line width=0.35pt, gray!60] (n1) -- (n3);
\node[draw, line width=0.4pt, fill=white, inner sep=0pt,
      minimum width=2.6mm, minimum height=1.5mm, rotate=-31]
  (adm13) at ($(n1)!0.3!(n3)$) {};
\node[font=\tiny, fill=orange!8, inner sep=0.6pt]
  at ([xshift=1.2mm, yshift=1.7mm]adm13) {$\tilde{Y}_{13}$};

\draw[line width=0.35pt, gray!60] (n2) -- (n4);
\node[draw, line width=0.4pt, fill=white, inner sep=0pt,
      minimum width=2.6mm, minimum height=1.5mm, rotate=31]
  (adm24) at ($(n2)!0.3!(n4)$) {};
\node[font=\tiny, fill=orange!8, inner sep=0.6pt]
  at ([xshift=-1.2mm, yshift=1.7mm]adm24) {$\tilde{Y}_{24}$};

\draw[line width=0.5pt] (n1) -- ++(0, 0.45cm);
\node[draw, line width=0.5pt, fill=white, inner sep=0pt,
      minimum width=1.8mm, minimum height=3mm]
  (adm1) at ([yshift=0.225cm]n1) {};
\node[font=\tiny, inner sep=1pt, right=0.8pt of adm1] {$\tilde{Y}_1$};
\coordinate (g1) at ([yshift=0.45cm]n1);
\draw[line width=0.5pt] ([xshift=-1.4mm]g1) -- ([xshift=1.4mm]g1);
\draw[line width=0.5pt] ([xshift=-0.9mm,yshift=0.4mm]g1) -- ([xshift=0.9mm,yshift=0.4mm]g1);
\draw[line width=0.5pt] ([xshift=-0.4mm,yshift=0.8mm]g1) -- ([xshift=0.4mm,yshift=0.8mm]g1);

\draw[line width=0.5pt] (n2) -- ++(0, 0.45cm);
\node[draw, line width=0.5pt, fill=white, inner sep=0pt,
      minimum width=1.8mm, minimum height=3mm]
  (adm2) at ([yshift=0.225cm]n2) {};
\node[font=\tiny, inner sep=1pt, left=0.8pt of adm2] {$\tilde{Y}_2$};
\coordinate (g2) at ([yshift=0.45cm]n2);
\draw[line width=0.5pt] ([xshift=-1.4mm]g2) -- ([xshift=1.4mm]g2);
\draw[line width=0.5pt] ([xshift=-0.9mm,yshift=0.4mm]g2) -- ([xshift=0.9mm,yshift=0.4mm]g2);
\draw[line width=0.5pt] ([xshift=-0.4mm,yshift=0.8mm]g2) -- ([xshift=0.4mm,yshift=0.8mm]g2);

\draw[line width=0.5pt] (n3) -- ++(0, -0.45cm);
\node[draw, line width=0.5pt, fill=white, inner sep=0pt,
      minimum width=1.8mm, minimum height=3mm]
  (adm3) at ([yshift=-0.225cm]n3) {};
\node[font=\tiny, inner sep=1pt, left=0.8pt of adm3] {$\tilde{Y}_3$};
\coordinate (g3) at ([yshift=-0.45cm]n3);
\draw[line width=0.5pt] ([xshift=-1.4mm]g3) -- ([xshift=1.4mm]g3);
\draw[line width=0.5pt] ([xshift=-0.9mm,yshift=-0.4mm]g3) -- ([xshift=0.9mm,yshift=-0.4mm]g3);
\draw[line width=0.5pt] ([xshift=-0.4mm,yshift=-0.8mm]g3) -- ([xshift=0.4mm,yshift=-0.8mm]g3);

\draw[line width=0.5pt] (n4) -- ++(0, -0.45cm);
\node[draw, line width=0.5pt, fill=white, inner sep=0pt,
      minimum width=1.8mm, minimum height=3mm]
  (adm4) at ([yshift=-0.225cm]n4) {};
\node[font=\tiny, inner sep=1pt, right=0.8pt of adm4] {$\tilde{Y}_4$};
\coordinate (g4) at ([yshift=-0.45cm]n4);
\draw[line width=0.5pt] ([xshift=-1.4mm]g4) -- ([xshift=1.4mm]g4);
\draw[line width=0.5pt] ([xshift=-0.9mm,yshift=-0.4mm]g4) -- ([xshift=0.9mm,yshift=-0.4mm]g4);
\draw[line width=0.5pt] ([xshift=-0.4mm,yshift=-0.8mm]g4) -- ([xshift=0.4mm,yshift=-0.8mm]g4);

\foreach \y in {\aA,\aB,\aC,\aD}{
  \draw[line width=0.7pt] (\xant,\y) -- ++(0,0.19);
  \draw[line width=0.7pt] (\xant-0.13,\y+0.32) -- (\xant,\y+0.19) -- (\xant+0.13,\y+0.32);
}
\node[font=\footnotesize] at (\xant,\aE) {$\vdots$};

\foreach \y in {\rA,\rB,\rC}{
  \draw[line width=0.7pt, fill=white, rounded corners=1pt]
    (\xusr-0.13,\y-0.20) rectangle (\xusr+0.13,\y+0.16);
  \draw[line width=0.5pt] (\xusr-0.075,\y-0.13) -- (\xusr+0.075,\y-0.13);
  \draw[line width=0.5pt] (\xusr-0.075,\y-0.03) -- (\xusr+0.075,\y-0.03);
  \draw[line width=0.7pt] (\xusr+0.05,\y+0.16) -- (\xusr+0.05,\y+0.34);
  \fill (\xusr+0.05,\y+0.34) circle (0.035);
}
\node[font=\footnotesize] at (\xusr,\rD) {$\vdots$};

\foreach \y/\a/\b/\c in {\rA/{$s_1$}/{$s_{q,1}$}/{$c_{q,1}$},
                         \rB/{$s_2$}/{$s_{q,2}$}/{$c_{q,2}$},
                         \rC/{$s_K$}/{$s_{q,K}$}/{$c_{q,K}$}}{
  \draw[flow] (\xs,\y) -- (\xdac-0.88,\y)
    node[midway, below=1.5pt, font=\scriptsize] {\a};
  \draw[flow] (\xdac+0.88,\y) -- (\xpa-0.50,\y)
    node[midway, below=1.5pt, font=\scriptsize] {\b};
  \draw[flow] (\xpa+0.50,\y) -- ([yshift=\y cm]milac.west)
    node[midway, below=1.5pt, font=\scriptsize] {\c};
}
\foreach \y/\lb in {\aA/{$x_{q,1}$}, \aB/{$x_{q,2}$}, \aC/{$x_{q,3}$}, \aD/{$x_{q,N}$}}{
  \draw[flow] ([yshift=\y cm]milac.east) -- (\xant,\y)
    node[pos=0.26, below=1.5pt, font=\scriptsize] {\lb};
}

\foreach \ay in {\aA,\aB,\aC,\aD}{
  \foreach \uy in {\rA,\rB,\rC}{
    \draw[line width=0.3pt, gray!40] (\xant+0.24,\ay) -- (\xusr-0.20,\uy);
  }
}
\node[font=\small\bfseries, fill=green!3, inner sep=2pt, rounded corners=1pt]
  at ($(\xant,1.42)!0.5!(\xusr,1.42)$) {$\Hm$};

\node[smallannot] (antlbl) at (\xant,-1.72) {Antenna array};
\node[smallannot] (usrlbl) at (\xusr,-1.72) {$K$ users};

\coordinate (annbase) at (0,-2.28);
\coordinate (sqcol) at ($(\xdac+0.88,0)!0.5!(\xpa-0.50,0)$);
\coordinate (cqcol) at ($(\xpa+0.50,0)!0.5!(milac.west)$);
\node[annot] at (sqcol |- annbase) {$\sv_q=\alpha_q\sv+\nv_s$};
\node[annot] at ([xshift=2.2mm]cqcol |- annbase) {$\cv_q=\operatorname{diag}(\sqrt{\pv})\,\sv_q$};
\node[annot] at ([xshift=2.4mm]milac.center |- annbase)
  {$\Fm=[\Thetam_{\mathrm{M}}]_{K+1:K+N,\,1:K}$};
\node[annot] at (\xant,0 |- annbase) {$\xv_q=\alpha_q\Wm\sv+\qv$};

\begin{scope}[on background layer]
  \path ([yshift= 1.9cm]fecenter)     coordinate (aft_t);
  \path ([yshift=-1.9cm]fecenter)     coordinate (aft_b);
  \path ([yshift= 1.9cm]milac.center) coordinate (ps_t);
  \path ([yshift=-1.9cm]milac.center) coordinate (ps_b);
  \coordinate (ch_t) at (\xant, 1.9);
  \coordinate (ch_b) at (\xant,-1.9);
  \coordinate (chR)  at (\xusr+0.35, 0);

  \node[draw, rounded corners=6pt, fill=blue!3, line width=0.8pt,
        inner xsep=8pt, inner ysep=0pt,
        fit=(dacA)(dacC)(paA)(paC)(aft_t)(aft_b),
        label={[zonelabel]above:Active Front End}] {};
  \node[draw, rounded corners=6pt, fill=orange!4, line width=0.8pt,
        inner xsep=6pt, inner ysep=0pt, fit=(milac)(ps_t)(ps_b),
        label={[zonelabel]above:Passive Network}] {};
  \node[draw, rounded corners=6pt, fill=green!3, line width=0.8pt,
        inner xsep=10pt, inner ysep=0pt,
        fit=(ch_t)(ch_b)(antlbl)(usrlbl)(chR),
        label={[zonelabel]above:Wireless Channel}] {};
\end{scope}

\end{tikzpicture}
\vspace{-0.3em}
\caption{Downlink MU-MISO MiLAC-aided beamforming architecture. The active front end consists of $K$ parallel \ac{RF} chains, each with a $b$-bit \ac{DAC} followed by a variable-gain \ac{PA} that sets the per-stream power $p_k$ and drives one \ac{MiLAC} input port. A lossless and reciprocal $(N+K)$-port \ac{MiLAC} (\ac{FC} tunable-admittance network) maps the $K$ stream ports to the $N$ antenna ports, yielding an effective stream-to-antenna mapping constrained by $\|\Fm\|_2 \le 1$. The $N$ antenna ports are fed passively, and the antenna array serves $K$ single-antenna users over the wireless channel $\Hm$. }
\label{fig:system}
\end{figure*}

\section{System Model}\label{sec:system_model}

\subsection{Quantization-Aware Signal Model}\label{sec:aqnm_model}

As illustrated in Fig.~\ref{fig:system}, we consider a downlink MU-MISO \ac{MiLAC}-aided beamforming system in which a \ac{BS} equipped with $N$ transmit antennas serves $K$ single-antenna users, where $K \le N$. The transmitter employs $K$ \ac{RF} chains to generate $K$ data streams, which are quantized by \acp{DAC}, amplified by their per-chain \acp{PA}, and then injected into the stream-side ports of a lossless and reciprocal $(N+K)$-port \ac{MiLAC}. The remaining $N$ ports are connected to the antenna array. Each \ac{RF} chain terminates in a dedicated \ac{PA} that drives one \ac{MiLAC} input port, so the $K$ active chains carry the full transmit power while the $(N+K)$-port network and its $N$ antenna ports remain entirely passive. The number of active devices is thus $K$, independent of the antenna number~\cite{Wu2026MiLAC,Nerini2025Capacity}.

Let
\begin{equation}
\cv=\operatorname{diag}\bigl(\sqrt{\pv}\bigr)\sv
\end{equation}
denote the power-loaded stream vector, where $\pv=[p_1,\ldots,p_K]^{\mathsf T}\succeq \zerov$ collects the per-stream powers and $\sv\sim\mathcal{CN}(\zerov,\Id_K)$ is the data symbol vector. Each stream passes through a $b$-bit \ac{DAC}, whose unit-power input $\sv$ is quantized under the \ac{AQNM}~\cite{Choi2022Energy} as
\begin{equation}
\sv_q = \alpha_q\sv + \nv_s,
\end{equation}
where $\nv_s\sim\mathcal{CN}(\zerov,\alpha_q\beta_q\Id_K)$ denotes the quantization noise and is uncorrelated with $\sv$. Here, $\alpha_q=1-\beta_q$ is the quantization gain, while $\beta_q$ is the normalized per-element distortion variance, tabulated for $b\le 5$ and approximated for $b>5$ as $\beta_q\approx \frac{\pi\sqrt{3}}{2}2^{-2b}$~\cite{Choi2022Energy}. The \ac{DAC} output is then amplified by the per-chain \ac{PA}, whose gain sets the transmit power of its stream, yielding
\begin{equation}
\cv_q = \operatorname{diag}\bigl(\sqrt{\pv}\bigr)\sv_q = \alpha_q\cv + \nv_c,
\end{equation}
where $\nv_c=\operatorname{diag}(\sqrt{\pv})\,\nv_s\sim\mathcal{CN}\bigl(\zerov,\alpha_q\beta_q\operatorname{diag}(\pv)\bigr)$ is uncorrelated with $\cv$.

The \ac{MiLAC} network maps $\cv_q$ linearly to the antenna domain as
\begin{equation}\label{eq:txq_milac}
\xv_q
=
\Fm \cv_q
=
\alpha_q \Fm \operatorname{diag}\bigl(\sqrt{\pv}\bigr)\sv + \Fm \nv_c,
\end{equation}
where $\Fm\in\mathbb{C}^{N\times K}$ denotes the stream-to-antenna mapping induced by the \ac{MiLAC}. Its microwave-circuit structure will be detailed in the next subsection.

Define the effective \ac{MiLAC}-aided beamformer as
\begin{equation}\label{eq:effective_beamformer}
\Wm=\Fm\operatorname{diag}\bigl(\sqrt{\pv}\bigr).
\end{equation}
Then \eqref{eq:txq_milac} becomes
\begin{equation}
\xv_q = \alpha_q \Wm \sv + \qv,
\end{equation}
where $\qv=\Fm\nv_c$ denotes the antenna-domain distortion. Its covariance is
\begin{equation}\label{eq:rq_milac}
\Rm_q
=
\mathbb{E}\bigl[\qv\qv^\HH\bigr]
=
\Fm\bigl(\alpha_q\beta_q\operatorname{diag}(\pv)\bigr)\Fm^\HH
=
\alpha_q\beta_q\Wm\Wm^\HH.
\end{equation}
Hence, $\Rm_q$ lies entirely in the \emph{column space} of $\Wm$. This structure is absent in conventional digital and hybrid architectures. In particular, digital beamforming with per-antenna \acp{DAC} typically yields a diagonal distortion covariance in the antenna domain, whereas hybrid beamforming does not in general lead to a covariance proportional to $\Wm\Wm^\HH$. 
This structured form of $\Rm_q$ will be instrumental in the exact row-space reduction developed in Section~\ref{sec:reformulation}.

\looseness=-1 The \acp{PA} deliver the stream signal into the network, so the post-quantized transmit power is measured at the \ac{MiLAC} input ports,
\begin{align}\label{eq:ptx_milac}
P_{\mathrm{tx}}^{\mathrm{milac}}\left(\pv\right)
&=
\mathbb{E}\bigl[\|\cv_q\|_2^2\bigr]=
\alpha_q^2\onev^{\mathsf T}\pv + \alpha_q\beta_q\onev^{\mathsf T}\pv\notag\\
&=\alpha_q\onev^{\mathsf T}\pv,
\end{align}
where the last equality follows from $\alpha_q^2+\alpha_q\beta_q=\alpha_q(\alpha_q+\beta_q)=\alpha_q$. The power radiated at the antenna ports is $\mathbb{E}[\|\xv_q\|_2^2]=\alpha_q\|\Wm\|_{\mathrm F}^2$, and $\|\Fm\|_2\le 1$ gives $\|\Wm\|_{\mathrm F}^2=\sum_{k}p_k\|\fv_k\|_2^2\le\onev^{\mathsf T}\pv$, with equality if and only if every excited column $\fv_k$ of $\Fm$ has unit norm. Any shortfall is reflected at the input ports of the lossless network and absorbed in the matched source impedances, and is therefore part of the power the \acp{PA} supply. \ac{EE} optimization thus favors beamformers that the network passes without reflection. By~\cite[Prop.~2]{Wu2026MiLAC}, a \ac{MiLAC}-feasible $\Wm$ attaining $\|\Wm\|_{\mathrm F}^2=\onev^{\mathsf T}\pv$ has mutually orthogonal columns, a condition that fully digital beamforming does not impose.

\begin{remark}[\ac{MiLAC}-Aided versus Hybrid Digital-\ac{MiLAC}]\label{rem:hybrid_milac}
The architecture studied in this paper is \ac{MiLAC}-aided beamforming, where the \ac{MiLAC} alone performs the stream-to-antenna mapping and no baseband digital beamformer is applied beyond power allocation. A hybrid digital-\ac{MiLAC} variant that inserts a $K\times K$ baseband beamformer before the \ac{MiLAC} can realize the full flexibility of digital beamforming with only $K$ \ac{RF} chains~\cite{Wu2026MiLAC}. We follow~\cite{Wu2026MiLAC} in focusing on the \ac{MiLAC}-aided architecture: a target digital beamformer can first be optimized using standard tools, and the corresponding \ac{MiLAC}-aided parameters can then be recovered via the \ac{SVD}-based construction of~\cite[Proposition~3]{Wu2026MiLAC}. The \ac{SE} loss of the resulting \ac{MiLAC}-aided beamformer relative to fully digital beamforming diminishes as the antenna-to-stream ratio $N/K$ grows~\cite{Fang2026Performance}, a regime that coincides with the large-array setting of primary interest in this paper. Under these conditions, \ac{MiLAC}-aided beamforming captures the essential architecture-level design space, and eliminating the baseband-mixing layer further brings the practical benefit of avoiding symbol-level digital processing at the transmitter~\cite{Wu2026MiLAC,Nerini2025AnalogII,Nerini2025Capacity}.
\end{remark}

\subsection{\ac{MiLAC} Circuit Model and Beamformer Feasible Set}\label{sec:circuit}

We next characterize the structure of $\Fm$ through microwave circuit analysis~\cite{Wu2026MiLAC,Nerini2025Capacity,Nerini2025AnalogII}. Let $\bar Y_n$ denote the shunt admittance from port $n$ to ground, and let $\bar Y_{n,m}$ denote the coupling admittance between ports $n$ and $m$. The corresponding multiport admittance matrix is given by
\begin{equation}
[\mathbf{Y}_{\mathrm c}]_{n,m}
=
\begin{cases}
-\bar Y_{n,m}, & n\neq m,\\[0.3em]
\bar Y_n+\sum_{j\neq n}\bar Y_{n,j}, & n=m.
\end{cases}
\end{equation}
Physically, $\mathbf{Y}_{\mathrm c}$ relates the vector of port currents to the vector of port voltages: its $(n,m)$ entry is the current entering port $n$ per unit voltage applied at port $m$ with all other ports short-circuited~\cite{Nerini2025AnalogI}. Reciprocity makes the couplings symmetric, $\bar Y_{n,m}=\bar Y_{m,n}$, so a \ac{FC} reciprocal network carries one tunable admittance per independent entry of the symmetric matrix $\mathbf{Y}_{\mathrm c}$: $N{+}K$ shunt and $(N{+}K)(N{+}K{-}1)/2$ coupling admittances~\cite{Wu2026MiLAC}, that is $N_{\mathrm{adm}}=(N{+}K)(N{+}K{+}1)/2$ in total, or $2346$ at the default $N=64$ and $K=4$ of Section~\ref{sec:results}.

The beamforming action is most naturally described in the scattering domain, where the scattering matrix specifies how incident waves at the ports are redistributed across the network. In particular, the stream-to-antenna block of the scattering matrix gives the linear mapping from the $K$ stream excitations to the $N$ antenna signals. With reference impedance $Z_0$ and reference admittance $Y_0=Z_0^{-1}$, the scattering matrix is
\begin{equation}\label{eq:F_structure}
\Thetam_{\mathrm M}
=
\left(\Id_{N+K}+Z_0\mathbf{Y}_{\mathrm c}\right)^{-1}
\left(\Id_{N+K}-Z_0\mathbf{Y}_{\mathrm c}\right).
\end{equation}
Losslessness makes $\Thetam_{\mathrm M}$ unitary, so the total outgoing wave power equals the total incident wave power, and reciprocity makes it symmetric, so the transmission between any two ports is the same in both directions~\cite{Wu2026MiLAC}.
The corresponding stream-to-antenna block is\footnote{Strictly speaking, a factor $1/2$ appears in front of the right-hand side. Following~\cite{Nerini2025Capacity,Wu2026MiLAC}, we absorb this voltage-division factor into the effective beamformer definition in \eqref{eq:effective_beamformer}.}
\begin{equation}
\Fm
=
\left[\Thetam_{\mathrm M}\right]_{K+1:K+N,\,1:K},
\end{equation}
where the subscript selects the rows $K{+}1$ to $K{+}N$ (the antenna ports) and the columns $1$ to $K$ (the stream ports) of $\Thetam_{\mathrm M}$.

The feasible-set characterization in~\cite{Wu2026MiLAC} shows that the lossless reciprocal hardware constraints admit a compact beamforming-space representation. For completeness, we restate the relevant result.

\begin{proposition}\label{prop:fullspace_geometry}
For a lossless reciprocal \ac{MiLAC}, a pair $(\Wm,\pv)$ is beamforming-feasible if and only if
\begin{equation}\label{eq:fullspace_geometry}
\Wm^\HH \Wm \preceq \operatorname{diag}(\pv).
\end{equation}
\end{proposition}

\begin{proof}
The result follows from Proposition~1 and Corollary~1 of~\cite{Wu2026MiLAC}. If $\Wm=\Fm\operatorname{diag}(\sqrt{\pv})$, then $\|\Fm\|_2\le 1$ because $\Fm$ is a submatrix of the unitary matrix $\Thetam_{\mathrm M}$ and is therefore contractive. Hence,
\begin{equation}
\Wm^\HH \Wm
=
\operatorname{diag}\bigl(\sqrt{\pv}\bigr)\Fm^\HH\Fm\operatorname{diag}\bigl(\sqrt{\pv}\bigr)
\preceq
\operatorname{diag}(\pv).
\end{equation}

Conversely, suppose that $(\Wm,\pv)$ satisfies~\eqref{eq:fullspace_geometry}. Define $\Fm=\Wm\operatorname{diag}(\sqrt{\pv^\dagger})$, where $\pv^\dagger$ denotes the elementwise pseudoinverse of $\pv$. Then $\Wm=\Fm\operatorname{diag}(\sqrt{\pv})$, and
\begin{equation}
\Fm^\HH\Fm
=
\operatorname{diag}\bigl(\sqrt{\pv^\dagger}\bigr)\Wm^\HH\Wm\operatorname{diag}\bigl(\sqrt{\pv^\dagger}\bigr)
\preceq
\Id_K,
\end{equation}
which implies $\|\Fm\|_2\le 1$. The pseudoinverse is benign when some entries of $\pv$ are zero: if $p_k=0$, then~\eqref{eq:fullspace_geometry} forces the $k$th column of $\Wm$ to be zero, so the factorization remains valid columnwise.
\end{proof}

The inequality in~\eqref{eq:fullspace_geometry} is the beamforming-space consequence of passive reciprocal hardware: \ac{MiLAC} can redistribute the injected stream excitations across the antenna array, but the resulting effective beamformer must remain compatible with the available stream powers.

\subsection{Quantization-Aware \ac{EE} Formulation}\label{sec:rate_power}

\subsubsection{Signal model and sum SE} Let $\hv_k\in\mathbb{C}^{N}$ denote the channel vector for user $k$. The received signal at user $k$ is
\begin{align}\label{eq:rx}
r_k = &\hv_k^\HH \xv_q + n_k\\
=&\alpha_q\hv_k^\HH \wv_k s_k + \underbrace{\sum_{j\neq k}\alpha_q\hv_k^\HH \wv_j s_j}_{\text{Multiuser interference}} + \underbrace{\hv_k^\HH\qv}_{\text{Distortion}}+ n_k,\notag
\end{align}
where $\wv_k\in\mathbb{C}^{N}$ denotes the $k$-th column of $\Wm$ and $n_k\sim\mathcal{CN}(0,\sigma^2)$ is the \ac{AWGN}.

Under the \ac{AQNM}, the resulting \ac{SINR} at user $k$ is
\begin{equation}\label{eq:sinr_milac}
\Gamma_k\left(\Wm\right)
=
\frac{\alpha_q^2|\hv_k^\HH \wv_k|^2}
{\alpha_q^2\sum_{j\neq k}|\hv_k^\HH \wv_j|^2 + \hv_k^\HH \Rm_q \hv_k + \sigma^2},
\end{equation}
and the corresponding quantization-aware sum \ac{SE} is
\begin{equation}\label{eq:rate_sum}
R\left(\Wm\right)
=
\sum_{k=1}^{K}\log_2\bigl(1+\Gamma_k\left(\Wm\right)\bigr).
\end{equation}

The total \ac{MiLAC} power consumption consists of the \ac{PA} supply power and the architecture-dependent circuit power.

\subsubsection{PA supply power} The \ac{PA} supply power is
\begin{equation}\label{eq:ppa}
P_{\mathrm{PA}}^{\mathrm{milac}}\left(\pv\right) = \frac{P_{\mathrm{tx}}^{\mathrm{milac}}\left(\pv\right)}{\eta_{\mathrm{PA}}},
\end{equation}
where $\eta_{\mathrm{PA}}$ denotes the \ac{PA} drain efficiency~\cite{Choi2022Energy,Bjornson2017Massive}. The $K$ \acp{PA} sit at the \ac{MiLAC} input ports, so $P_{\mathrm{PA}}^{\mathrm{milac}}$ is the supply drawn to deliver $P_{\mathrm{tx}}^{\mathrm{milac}}$ into the network. For the digital and hybrid benchmarks the \acp{PA} drive the antenna ports, where the delivered and radiated powers coincide. In every case the budget and the supply refer to the same quantity, namely the power delivered by the \acp{PA}.

\subsubsection{Static power} The pre-\ac{MiLAC} circuit power associated with the $K$ \ac{RF} chains is modeled as~\cite{Choi2022Energy}
\begin{equation}\label{eq:milac_circuit}
P_{\mathrm{stat}}^{\mathrm{circ}}
=
P_{\mathrm{LO}} + K\bigl(P_{\mathrm{RF}} + P_{\mathrm{DAC,pair}}\bigr),
\end{equation}
where $P_{\mathrm{LO}}$ is the local-oscillator power. The per-chain \ac{RF} power is decomposed as~\cite{Choi2022Energy}
\begin{equation}\label{eq:prf}
P_{\mathrm{RF}} = 2P_{\mathrm{LP}}+2P_{\mathrm{M}}+P_{\mathrm{H}},
\end{equation}
where $P_{\mathrm{LP}}$, $P_{\mathrm{M}}$, and $P_{\mathrm{H}}$ denote the powers consumed by the low-pass filters, mixers, and 90$^\circ$ hybrid with buffer, respectively. The power consumed by each \ac{DAC} pair is~\cite{Ribeiro2018Energy}
\begin{equation}\label{eq:pdac}
P_{\mathrm{DAC,pair}}
=
2\left(1.5\times 10^{-5}2^b + 9\times 10^{-12}f_s b\right),
\end{equation}
where $f_s$ denotes the system sampling rate. The bracketed term is the power of one binary-weighted current-steering \ac{DAC}, modeled as a function of its resolution and sampling rate~\cite{Cui2005Energy,Ribeiro2018Energy}, and the factor of two counts the two \acp{DAC} of one chain, as in the power models of~\cite{Ribeiro2018Energy}.

We model the static \ac{MiLAC} network power as
\begin{equation}
P_{\mathrm{stat}}^{\mathrm{milac}} = N_{\mathrm{adm}}\,P_{\mathrm{drv}},
\end{equation}
where $N_{\mathrm{adm}} = (N+K)(N+K+1)/2$
is the number of tunable admittances in a \ac{FC} reciprocal network~\cite{Shen2022Modeling,Zhou2024Optimizing}\footnote{Stem-connected architectures lower the admittance number from $\mathcal{O}((N{+}K)^2)$ to $\mathcal{O}(NK)$ with no significant sum-rate loss~\cite{Nerini2025Reduced,Zhang2026StemBF}. The saving grows with the array, from $4.5\times$ at $N=64$ to $16.5\times$ at $N=256$ for $K=4$, and reduces $P_{\mathrm{stat}}^{\mathrm{milac}}$ by the same factor. A sparser network, however, has fewer spare degrees of freedom to compensate for coarse admittance resolution, so it may need more resolution bits, and hence more drive power, per admittance. In practice, the net power saving therefore depends on both the admittance number and the required resolution. We restrict attention to the \ac{FC} topology here and leave the \ac{EE} comparison between the \ac{FC} and stem-connected topologies to future work.} and $P_{\mathrm{drv}}$ denotes the power drawn by the drive circuit of one tunable admittance. Since the admittances are lossless, they dissipate no power themselves, and this term accounts entirely for the circuitry that holds each of them at its assigned value.

\subsubsection{Total power and EE} The total power consumption is
\begin{equation}\label{eq:pmilac_tot}
P_{\mathrm{tot}}^{\mathrm{milac}}\left(\pv\right)
=
\frac{P_{\mathrm{tx}}^{\mathrm{milac}}\left(\pv\right)}{\eta_{\mathrm{PA}}}
+ P_{\mathrm{stat}}^{\mathrm{circ}}
+ P_{\mathrm{stat}}^{\mathrm{milac}}.
\end{equation}

Finally, the \ac{MiLAC} \ac{EE} is defined as
\begin{equation}\label{eq:ee_def}
\operatorname{EE}
=
\frac{B R\left(\Wm\right)}{P_{\mathrm{tot}}^{\mathrm{milac}}\left(\pv\right)}
\quad \text{(bit/J)},
\end{equation}
where $B$ denotes the system bandwidth.

Four modeling choices define the scope of this framework. First, \ac{CSI} is assumed perfect. Channel estimation schemes for \ac{MiLAC}-aided systems, operating in the analog domain, have recently become available~\cite{Zhang2026Channel,Zhang2026ChannelBF}, and quantifying the \ac{EE} cost of imperfect \ac{CSI} remains open. Second, the admittances are lossless. Nonzero conductance dissipates part of the injected power and changes the realized stream-to-antenna mapping. For the closely related beyond-diagonal \ac{RIS} (BD-RIS), built from the same kind of tunable admittance network, component-level loss models and the resulting performance degradation have been established in~\cite{Peng2026LossyBDRIS}. Third, antenna mutual coupling is neglected. Since a \ac{MiLAC} processes signals at \ac{RF}, coupling also changes the realized stream-to-antenna mapping, and a physics-compliant model is available in~\cite{Nerini2026Physics}. Fourth, the admittances are optimized as continuously tunable. Section~\ref{sec:setup} discusses the finite-resolution implementation and its drive power.

\section{Problem Formulation and Solution Structure for Quantization-Aware \ac{EE} Optimization}\label{sec:reformulation}

\subsection{Quantization-Aware \ac{EE} Optimization Problem}\label{sec:full_problem}

Combining Proposition~\ref{prop:fullspace_geometry}, the transmit-power expression in~\eqref{eq:ptx_milac}, and the \ac{EE} definition in~\eqref{eq:ee_def}, we obtain the quantization-aware \ac{EE} optimization problem
\begin{subequations}\label{eq:problem_full}
\begin{alignat}{2}
    &\max_{\Wm,\pv}\quad &&
    \frac{B R\left(\Wm\right)}{P_{\mathrm{tot}}^{\mathrm{milac}}\left(\pv\right)} \label{eq:problem_full_obj}\\
    &\text{s.t.}\quad &&
    \Wm^\HH \Wm \preceq \operatorname{diag}\left(\pv\right), 
     \label{eq:problem_full_a}\\
    & &&
    \alpha_q\onev^{\mathsf T}\pv \le P_{\max}, \label{eq:problem_full_b}
\end{alignat}
\end{subequations}
where $P_{\max}$ denotes the maximum transmit power.

Problem~\eqref{eq:problem_full} is challenging because of its nonconvex fractional objective. Moreover, unlike conventional digital beamforming, where the sum-power constraint acts directly on the beamformer, here the budget applies to the power injected at the stream ports while the rate depends on the antenna-domain beamformer $\Wm$. The two are linked only through the feasibility condition in~\eqref{eq:problem_full_a}, so $\Wm$ and $\pv$ are coupled: raising $p_k$ enlarges the feasible set for $\Wm$ but increases the transmit power in the objective. Although the matrix inequality $\Wm^\HH \Wm \preceq \operatorname{diag}(\pv)$ is convex and can be represented as a \ac{LMI}, a direct treatment would generally lead to a \ac{SDP}-based formulation, which becomes computationally burdensome for large-scale arrays. Motivated by~\cite{Wu2026MiLAC}, we therefore seek an exact reduced-dimension reformulation of~\eqref{eq:problem_full}, which will form the basis of the low-complexity algorithm developed in the next section.

\subsection{Reduced-Dimension Reformulation}\label{sec:range_restrict}
Two reduced $K\times K$ variables are used in the rest of the paper, and we introduce them together here so that their roles are clear. This section replaces the $N\times K$ effective beamformer $\Wm$ by the variable $\Ym$, related to $\Wm$ by $\Wm=\Hm^\HH\bar{\Hm}^{-1/2}\Ym\operatorname{diag}(\sqrt{\pv})$ in Proposition~\ref{prop:reduced_param}, where $\bar{\Hm}=\Hm\Hm^\HH$, and the \ac{EE} algorithm of Section~\ref{sec:lc_solver} works entirely with $(\Ym,\pv)$. Section~\ref{sec:continuation} uses a second variable $\Xm$, related to $\Wm$ by $\Wm=\Hm^\HH\Xm$ and chosen because it makes the weighted-sum tradeoff objective amenable to convex approximation.

\looseness=-1 Let $\Hm\in\mathbb{C}^{K\times N}$ denote the channel matrix, whose $k$th row is $\hv_k^\HH$.
Define the orthogonal projector\footnote{We assume $K \le N$ and $\operatorname{rank}(\Hm)=K$, so that $\Hm\Hm^\HH$ is invertible. This condition is typically satisfied in practice. If $\operatorname{rank}(\Hm)<K$, the same development carries over by replacing the inverse with the Moore--Penrose pseudoinverse, which is omitted here for brevity.}
\begin{equation}\label{eq:proj}
    \Pim_{\Hm} = \Hm^\HH(\Hm\Hm^\HH)^{-1}\Hm.
\end{equation}
The following lemma shows that the feasible set of $\Wm$ can be restricted to $\operatorname{Ran}(\Hm^\HH)$, the column space of $\Hm^\HH$, without loss of optimality.

\begin{lemma}\label{lem:range_restrict}
For any feasible pair $(\Wm,\pv)$ of~\eqref{eq:problem_full}, let $\Wm_{\parallel}=\Pim_{\Hm}\Wm$ denote the projection of $\Wm$ onto $\operatorname{Ran}(\Hm^\HH)$. Then $(\Wm_{\parallel},\pv)$ is feasible, satisfies $\Hm\Wm_{\parallel}=\Hm\Wm$, and attains the same transmit power
\begin{equation}\label{eq:power_projection}
    P_{\mathrm{tx}}^{\mathrm{milac}}(\pv)
    =
    \alpha_q\onev^{\mathsf T}\pv .
\end{equation}
Hence, an optimum of~\eqref{eq:problem_full} always exists with $\Wm \in \operatorname{Ran}(\Hm^\HH)$.
\end{lemma}

\begin{proof}
Decompose $\Wm$ as
\begin{equation}
\Wm=\Pim_{\Hm}\Wm+\left(\Id-\Pim_{\Hm}\right)\Wm.
\end{equation}
The first term lies in $\operatorname{Ran}(\Hm^\HH)$, whereas the second lies in $\operatorname{Null}(\Hm)$. Since the null-space component is annihilated by $\Hm$, replacing $\Wm$ with $\Wm_{\parallel}=\Pim_{\Hm}\Wm$ leaves $\Hm\Wm$ unchanged and therefore preserves the desired signal and multiuser interference terms. The transmit power depends on $\pv$ alone and is therefore unaffected by the projection. Finally, feasibility with respect to~\eqref{eq:problem_full_a} is preserved because
\begin{equation}
\Wm_{\parallel}^\HH \Wm_{\parallel}
\preceq
\Wm^\HH\Wm
\preceq
\operatorname{diag}\left(\pv\right),
\end{equation}
which completes the proof.
\end{proof}

\begin{remark}[\ac{MiLAC}-Aided Beamforming Structural Advantage]\label{rem:rowspace_advantage}
The range restriction in Lemma~\ref{lem:range_restrict} relies on a structural feature that is specific to quantized \ac{MiLAC}-aided beamforming. Since the \acp{DAC} operate in the stream domain, the distortion observed at user $k$ is
\begin{equation}\label{eq:distortion_milac}
    \hv_k^\HH \Rm_q \hv_k
    =
    \alpha_q\beta_q \sum_{j=1}^{K} |C_{k,j}|^2,
\end{equation}
where $\Cm=\Hm\Wm\in\mathbb{C}^{K\times K}$ is the effective coupling matrix with $(k,j)$th entry $C_{k,j}$. Thus, the distortion depends only on $\Cm$, which is preserved under the projection $\Wm \mapsto \Pim_{\Hm}\Wm$. Consequently, the useful signal, the multiuser interference, and the quantization distortion remain unchanged, while the transmit power can only decrease.

This property does \emph{not} generally hold for digital or hybrid beamforming benchmarks. In digital beamforming, quantization is applied independently to the $N$ antenna-domain signals, which yields a diagonal distortion covariance of the form $\Rm_{q,\mathrm{dig}}=\alpha_q\beta_q\operatorname{diag}(\Wm\Wm^\HH)$. Although the projection preserves $\Hm\Wm$, it generally alters the row norms of $\Wm$ and therefore changes the user-side distortion. A similar obstruction arises in hybrid beamforming. As a result, \ac{MiLAC}-aided beamforming admits an exact $K\times K$ reduced-dimension representation regardless of \ac{DAC} resolution.
\end{remark}

Lemma~\ref{lem:range_restrict} shows that the effective beamformer in any optimal feasible pair can be confined to $\operatorname{Ran}(\Hm^\HH)$. We next combine this structural result with Proposition~\ref{prop:fullspace_geometry} to derive an exact reduced-dimension parameterization. Let $\bar{\Hm}=\Hm\Hm^\HH\in\mathbb{C}^{K\times K}$ denote the channel Gram matrix.

\begin{proposition}\label{prop:reduced_param}
Assume that $\Hm$ has full row rank. For any pair $(\Wm,\pv)$ satisfying~\eqref{eq:fullspace_geometry} with $\Wm \in \operatorname{Ran}(\Hm^\HH)$, there exists $\Ym \in \mathbb{C}^{K\times K}$ such that
\begin{equation}\label{eq:milac_param}
    \Wm
    =
    \Hm^\HH \bar{\Hm}^{-1/2}\Ym \operatorname{diag}\left(\sqrt{\pv}\right),
    \quad
    \left\|\Ym\right\|_2 \le 1.
\end{equation}
Conversely, any pair $(\Ym,\pv)$ satisfying $\|\Ym\|_2\le 1$ and $\pv \succeq \zerov$ generates, through~\eqref{eq:milac_param}, a pair $(\Wm,\pv)$ that satisfies~\eqref{eq:fullspace_geometry}.
\end{proposition}

\begin{proof}
Since $\Wm \in \operatorname{Ran}(\Hm^\HH)$, there exists $\Xm \in \mathbb{C}^{K\times K}$ such that $\Wm=\Hm^\HH\Xm$. Define
\begin{equation}
\Ym=\bar{\Hm}^{1/2}\Xm \operatorname{diag}(\sqrt{\pv^\dagger}),
\end{equation}
where $\pv^\dagger$ denotes the elementwise pseudoinverse of $\pv$. If $p_k=0$, then~\eqref{eq:fullspace_geometry} implies $\wv_k=\zerov$. Since $\Hm^\HH$ has full column rank, this further implies $\xv_k=\zerov$, so the above definition remains well posed even for zero-power columns. Rearranging gives
\begin{equation}
\Xm=\bar{\Hm}^{-1/2}\Ym \operatorname{diag}\left(\sqrt{\pv}\right),
\end{equation}
and hence~\eqref{eq:milac_param}. Moreover,
\begin{equation}
\begin{aligned}
\Ym^\HH\Ym
&=
\operatorname{diag}(\sqrt{\pv^\dagger})\Xm^\HH\bar{\Hm}\Xm
\operatorname{diag}(\sqrt{\pv^\dagger}) \\
&=
\operatorname{diag}(\sqrt{\pv^\dagger})\Wm^\HH\Wm
\operatorname{diag}(\sqrt{\pv^\dagger})
\preceq \Id_K,
\end{aligned}
\end{equation}
where the last step follows from~\eqref{eq:fullspace_geometry}. Therefore, $\|\Ym\|_2\le 1$.

Conversely, let $\Wm$ be generated by~\eqref{eq:milac_param} with $\|\Ym\|_2\le 1$. Then we have
\begin{equation}
\Wm^\HH\Wm
=
\operatorname{diag}\left(\sqrt{\pv}\right)\Ym^\HH\Ym\operatorname{diag}\left(\sqrt{\pv}\right)
\preceq
\operatorname{diag}\left(\pv\right),
\end{equation}
which is exactly~\eqref{eq:fullspace_geometry}.
\end{proof}

By Proposition~\ref{prop:reduced_param}, \eqref{eq:problem_full} can be equivalently rewritten as
\begin{subequations}\label{eq:problem_reduced}
\begin{alignat}{2}
    &\max_{\Ym,\pv}\quad &&
    \frac{B R\left(\Ym,\pv\right)}{P_{\mathrm{tot}}^{\mathrm{milac}}\left(\pv\right)} \label{eq:problem_reduced_obj}\\
    &\text{s.t.}\quad &&
    \|\Ym\|_2 \le 1, \label{eq:problem_reduced_a}\\
    & &&
    \alpha_q\onev^{\mathsf T}\pv \le P_{\max}, \label{eq:problem_reduced_b}
\end{alignat}
\end{subequations}
where the parameterization $\Wm = \Hm^\HH \bar{\Hm}^{-1/2}\Ym \operatorname{diag}(\sqrt{\pv})$ in~\eqref{eq:milac_param} yields the following reduced-coordinate identities as
\begin{equation}\label{eq:coupling}
    \Cm = \Hm\Wm = \bar{\Hm}^{1/2}\Ym \operatorname{diag}\left(\sqrt{\pv}\right),
\end{equation}
and
\begin{equation}\label{eq:ptx_reduced}
    P_{\mathrm{tx}}^{\mathrm{milac}}\left(\pv\right)
    =
    \alpha_q\onev^{\mathsf T}\pv .
\end{equation}

\begin{proposition}\label{prop:exact_equiv}
Problem~\eqref{eq:problem_reduced} is an exact reformulation of the original quantization-aware \ac{EE} optimization problem~\eqref{eq:problem_full} for \ac{MiLAC}-aided beamforming. In particular, the passage from the full-space variables $(\Wm,\pv)$ to the reduced coordinates $(\Ym,\pv)$ introduces neither surrogate rate expressions nor surrogate power quantities.
\end{proposition}

\begin{proof}
By Lemma~\ref{lem:range_restrict}, an optimum of~\eqref{eq:problem_full} can be taken with $\Wm \in \operatorname{Ran}(\Hm^\HH)$. Proposition~\ref{prop:reduced_param} then establishes a one-to-one correspondence between feasible points in that subspace and feasible reduced-variable pairs $(\Ym,\pv)$ through~\eqref{eq:milac_param}. The useful signal and interference terms depend only on $\Hm\Wm$, which is reproduced exactly by~\eqref{eq:coupling}, while the post-quantized transmit power is preserved exactly through~\eqref{eq:ptx_reduced}. Therefore, both the objective and the constraints are transferred without approximation.
\end{proof}

\section{Low-Complexity Solution for Quantization-Aware \ac{EE} Optimization}\label{sec:lc_solver}

\subsection{Dinkelbach Reformulation}\label{sec:dinkelbach}

To handle the fractional objective in~\eqref{eq:problem_reduced}, we adopt the Dinkelbach method~\cite{Dinkelbach1967}. This yields a double-loop algorithm. In the outer loop, the Dinkelbach parameter $\lambda$ is updated according to the current \ac{EE} value. In the inner loop, for a given $\lambda$, we solve the corresponding subtractive-form problem in $(\Ym,\pv)$.

For any feasible $(\Ym,\pv)$, the total power can be written as
\begin{equation}\label{eq:ptot_split}
P_{\mathrm{tot}}^{\mathrm{milac}}\left(\pv\right)
=
P_{\mathrm{st}}^{\mathrm{milac}}
+
\frac{P_{\mathrm{tx}}^{\mathrm{milac}}\left(\pv\right)}{\eta_{\mathrm{PA}}},
\end{equation}
where $P_{\mathrm{st}}^{\mathrm{milac}}$ collects all static, signal-independent power terms. Therefore, for a given $\lambda \ge 0$, the Dinkelbach inner problem becomes
\begin{subequations}\label{eq:dinkelbach}
\begin{alignat}{2}
    &\max_{\Ym,\pv}\quad &&
    B R\left(\Ym,\pv\right)
    -
    \gamma(\lambda) P_{\mathrm{tx}}^{\mathrm{milac}}\left(\pv\right) \label{eq:dk_obj}\\
    &\text{s.t.}\quad &&
    \|\Ym\|_2 \le 1, \label{eq:dk_a}\\
    & &&
    \alpha_q\onev^{\mathsf T}\pv \le P_{\max}, \label{eq:dk_b}
\end{alignat}
\end{subequations}
where $\gamma(\lambda)=\frac{\lambda}{\eta_{\mathrm{PA}}}$.
The constant term $\lambda P_{\mathrm{st}}^{\mathrm{milac}}$ is omitted from~\eqref{eq:dinkelbach}, since it does not affect the inner-loop optimizer.

After solving~\eqref{eq:dinkelbach}, the outer-loop parameter is updated as~\cite{Dinkelbach1967}
\begin{equation}\label{dinkel_para}
\lambda
=
\frac{B R\left(\Ym,\pv\right)}{P_{\mathrm{tot}}^{\mathrm{milac}}\left(\pv\right)},
\end{equation}
with $R(\Ym,\pv)$ and $P_{\mathrm{tot}}^{\mathrm{milac}}(\pv)$ evaluated at the current inner-loop solution.

Problem~\eqref{eq:dinkelbach} remains nonconvex because $\Ym$ and $\pv$ are coupled in both the rate term and the transmit-power term. We therefore solve it by \ac{WMMSE}-based alternating optimization.

\subsection{WMMSE Solution to the Inner Problem}\label{sec:wmmse}

Define
\begin{equation}\label{eq:Z_def}
\Zm = \bar{\Hm}^{1/2}\Ym,
\end{equation}
so that, by~\eqref{eq:coupling},
\begin{equation}\label{eq:C_from_Z}
\Cm=\Zm\operatorname{diag}(\sqrt{\pv}).
\end{equation}
Let $u_k\in\mathbb{C}$ denote the scalar receive filter for user $k$. The quantization-aware mean-square error at user $k$ is
\begin{equation}\label{eq:mse}
\begin{aligned}
e_k\left(\Ym,\pv\right)
&=
\left|1-\alpha_q u_k C_{k,k}\right|^2
+\alpha_q^2|u_k|^2\sum_{j\neq k}|C_{k,j}|^2 \\
&\quad
+\alpha_q\beta_q|u_k|^2\sum_{j=1}^{K}|C_{k,j}|^2
+\sigma^2|u_k|^2.
\end{aligned}
\end{equation}

Introduce positive weights $\omega_k$, and collect them together with $u_k$ into $\uv=[u_1,\ldots,u_K]^\TT$ and $\omegav=[\omega_1,\ldots,\omega_K]^\TT$.
Using the standard \ac{WMMSE} identity with natural logarithm~\cite{Shi2011Iteratively},
\begin{equation}\label{eq:wmmse_identity}
\min_{u_k,\omega_k>0}\bigl(\omega_k e_k\left(\Ym,\pv\right)-\ln\omega_k\bigr)
=
1-\ln\left(1+\Gamma_k\left(\Ym,\pv\right)\right),
\end{equation}
we transform \eqref{eq:dinkelbach} equivalently into
\begin{subequations}\label{eq:wmmse_problem}
\begin{alignat}{2}
&\min_{\Ym,\pv,\uv,\omegav}\quad &&
\frac{B}{\ln 2}\sum_{k=1}^{K}\bigl(\omega_k e_k\left(\Ym,\pv\right)-\ln\omega_k\bigr)
+\gamma(\lambda)P_{\mathrm{tx}}^{\mathrm{milac}}\left(\pv\right) \notag\\
&\text{s.t.}\quad &&
\|\Ym\|_2\le 1, \label{eq:w_a}\\
& &&
\alpha_q\onev^{\mathsf T}\pv\le P_{\max}. \label{eq:w_b}
\end{alignat}
\end{subequations}
The additive constant $-BK/\ln 2$ has been omitted, since it does not affect the optimizer.

We solve~\eqref{eq:wmmse_problem} by alternating minimization over the four variable blocks $(\uv,\omegav,\pv,\Ym)$.

\subsubsection{Update of $u_k$ and $\omega_k$}

For fixed $(\Ym,\pv)$, the problem decouples across users. Minimizing~\eqref{eq:mse} with respect to $u_k$ gives the quantization-aware MMSE receiver
\begin{equation}\label{eq:u_update}
u_k
=
\frac{\alpha_q C_{k,k}^{\ast}}
{\alpha_q\sum_{j=1}^{K}|C_{k,j}|^2+\sigma^2},
\end{equation}
where we used $\alpha_q^2+\alpha_q\beta_q=\alpha_q$. Then, for fixed $u_k$, minimizing with respect to $\omega_k>0$ yields
\begin{equation}\label{eq:omega_update}
\omega_k=\frac{1}{e_k\left(\Ym,\pv\right)},
\end{equation}
where $e_k(\Ym,\pv)$ is evaluated using the updated $u_k$.

\subsubsection{Update of $\pv$}

For fixed $(\Ym,\uv,\omegav)$, and since $C_{i,k}=Z_{i,k}\sqrt{p_k}$, all $\pv$-dependent terms share a common positive factor $\alpha_q$, which can be dropped without affecting the minimizer. Define
\begin{align}
a_k &= \frac{B}{\ln 2}\,\omega_k\Re\{u_k Z_{k,k}\}, \label{eq:a_k}\\
b_k &= \frac{B}{\ln 2}\sum_{i=1}^{K}\omega_i|u_i|^2|Z_{i,k}|^2. \label{eq:b_k}
\end{align}
Then the $\pv$-subproblem becomes
\begin{subequations}\label{eq:p_subprob}
\begin{alignat}{2}
&\min_{\pv}\quad &&
\sum_{k=1}^{K}
\left[
\bigl(b_k+\gamma(\lambda)\bigr)p_k
-2a_k\sqrt{p_k}
\right] \label{eq:p_subprob_obj}\\
&\text{s.t.}\quad &&
\alpha_q\sum_{k=1}^{K} p_k \le P_{\max}. \label{eq:p_subprob_c}
\end{alignat}
\end{subequations}
The update is obtained from the \ac{KKT} conditions as
\begin{equation}\label{eq:p_update}
p_k
=
\left[
\frac{a_k}
{b_k+\gamma(\lambda)+\mu_p}
\right]_+^2,
\end{equation}
where $\mu_p\ge 0$ is the dual variable associated with~\eqref{eq:p_subprob_c}. Here, $\mu_p$ is understood as the rescaled multiplier after absorbing the factor $\alpha_q$ from the constraint. It can be found by bisection, and $\mu_p=0$ if the unconstrained solution already satisfies~\eqref{eq:p_subprob_c}.

\subsubsection{Update of $\Ym$}

For fixed $(\pv,\uv,\omegav)$, all $\Ym$-dependent terms in the objective share a common positive factor $\alpha_q$, which can again be dropped. The transmit-power budget~\eqref{eq:w_b} constrains $\pv$ alone, so the only constraint on $\Ym$ is the spectral-norm cap $\|\Ym\|_2 \le 1$ inherited from Proposition~\ref{prop:reduced_param}. The resulting $\Ym$-subproblem is
\begin{subequations}\label{Y_sub_prob}
\begin{alignat}{2}
&\min_{\Ym}\quad && f(\Ym) \label{eq:Y_sub_prob_obj}\\
&\text{s.t.}\quad && \|\Ym\|_2 \le 1. \label{eq:Y_sub_prob_c}
\end{alignat}
\end{subequations}
where
\begin{equation}
    f\left(\Ym\right) 
    =
    \operatorname{tr}\Big(\Ym^\HH \Qm \Ym \operatorname{diag}(\pv)\Big)
    -
    2\Re\left\{
    \operatorname{tr}\Big(\Lm \Ym \operatorname{diag}(\sqrt{\pv})\Big)
    \right\},\notag
\end{equation}
with
\begin{align}
\Qm
&=
\bar{\Hm}^{1/2}\operatorname{diag}(\zv)\bar{\Hm}^{1/2}, \label{eq:Qmat}\\
\zv
&=
\frac{B}{\ln 2}\,\omegav\odot|\uv|^{\odot 2}, \label{eq:zvec}\\
\Lm
&=
\operatorname{diag}\left(\frac{B}{\ln 2}\,\omegav\odot\uv\right)\bar{\Hm}^{1/2}. \label{eq:Lmat}
\end{align}
$\Qm$ captures the weighted interference-plus-distortion contribution induced by the \ac{WMMSE} variables.

To solve~\eqref{Y_sub_prob}, we adopt \ac{PGD}. The Wirtinger gradient of $f(\Ym)$ is
\begin{equation}\label{eq:gradY}
\nabla_{\Ym}f\left(\Ym\right)
=
2\Qm\Ym\operatorname{diag}\left(\pv\right)
-
2\Lm^\HH\operatorname{diag}\left(\sqrt{\pv}\right).
\end{equation}
The feasible set of~\eqref{Y_sub_prob} is the spectral-norm ball $\mathcal{C}=\{\Ym:\|\Ym\|_2\le 1\}$, which is closed, convex, and compact. Each \ac{PGD} step with step size $\tau>0$ takes a gradient step and projects onto $\mathcal{C}$. Let the \ac{SVD} of the gradient step be
\begin{equation}\label{eq:Y_svd}
\Ym-\tau\nabla_{\Ym}f
=
\Um\boldsymbol{\Sigma}\Vm^\HH,
\end{equation}
where $\boldsymbol{\Sigma}=\operatorname{diag}(\sigma_1,\ldots,\sigma_K)$. The Euclidean projection onto $\mathcal{C}$ clips every singular value to one as
\begin{equation}\label{eq:Ystep}
\Ym
\leftarrow
\Pi_{\mathcal{C}}\bigl(\Ym-\tau\nabla_{\Ym}f\bigr)
=
\Um\,\min(\boldsymbol{\Sigma},\Id_K)\,\Vm^\HH.
\end{equation}

\begin{proposition}[Convergence of the $\Ym$-update]\label{prop:Y_update_convergence}
For fixed $(\pv,\uv,\omegav)$, the objective $f(\Ym)$ in~\eqref{Y_sub_prob} is convex, and its Wirtinger gradient is Lipschitz continuous under the Frobenius norm with constant
\begin{equation}\label{eq:Lf}
L_f = 2\left\|\Qm\right\|_2\left\|\operatorname{diag}\left(\pv\right)\right\|_2.
\end{equation}
For any $\tau\in(0,1/L_f]$, the update~\eqref{eq:Ystep} converges to a global minimizer of $f(\Ym)$ over $\mathcal{C}$. In particular, the compactness of $\mathcal{C}$ guarantees existence of at least one such minimizer.
\end{proposition}

\begin{algorithm}[t]
\caption{Low-Complexity Inner Solution}
\label{algo:inner}
\footnotesize
\begin{algorithmic}[1]
\State \textbf{Input:} $\bar{\Hm}$, $P_{\max}$, noise variance $\sigma^2$, Dinkelbach coefficient $\gamma(\lambda)$, bandwidth $B$, tolerance $\epsilon_{\mathrm{in}} > 0$.
\State Initialize $(\Ym^{(0)},\pv^{(0)})$ from a feasible starting point.
\Repeat
\State Update MMSE receivers $\uv$ via~\eqref{eq:u_update} and weights $\omegav$ via~\eqref{eq:omega_update}.
\State Update $\pv$ via~\eqref{eq:p_update} with bisection on $\mu_p$.
\State Form $\Qm$, $\Lm$ from~\eqref{eq:Qmat}--\eqref{eq:Lmat} and solve the $\Ym$-subproblem~\eqref{Y_sub_prob} to convergence via \ac{PGD}~\eqref{eq:Ystep}.
\State Reconstruct $\Wm$ from~\eqref{eq:milac_param} and reproject only if residual numerical power violation remains.
\Until{$|g^{(n)} - g^{(n-1)}|/|g^{(n-1)}| \le \epsilon_{\mathrm{in}}$, where $g^{(n)}$ is the \ac{WMMSE} objective~\eqref{eq:wmmse_problem} at iteration $n$}
\State \textbf{Return} $(\Ym,\pv)$.
\end{algorithmic}
\end{algorithm}

\begin{algorithm}[t]
\caption{Quantization-Aware \ac{EE} Optimization Algorithm for \ac{MiLAC}-Aided MU-MISO Beamforming}
\label{algo:ee_overall}
\footnotesize
\begin{algorithmic}[1]
\State \textbf{Input:} channel $\Hm$, power budget $P_{\max}$, \ac{MiLAC} power parameters, tolerance $\epsilon_{\mathrm{out}}>0$.
\State Compute $\bar{\Hm}=\Hm\Hm^\HH$ and its matrix square root $\bar{\Hm}^{1/2}$. 
\State Initialize $(\Ym^{(0)},\pv^{(0)})$ to a feasible point and set $\lambda^{(0)} \leftarrow B R(\Ym^{(0)},\pv^{(0)}) / P_{\mathrm{tot}}^{\mathrm{milac}}(\pv^{(0)})$.
\For{$n=1,2,\ldots$}
\State Solve the inner problem~\eqref{eq:dinkelbach} via Algorithm~\ref{algo:inner} and obtain $(\Ym,\pv)$.
\State Update $\lambda^{(n)} \leftarrow B R(\Ym,\pv) \big/ P_{\mathrm{tot}}^{\mathrm{milac}}(\pv)$.
\If{$|\lambda^{(n)} - \lambda^{(n-1)}|/\lambda^{(n-1)} \le \epsilon_{\mathrm{out}}$}
\State \textbf{break}
\EndIf
\EndFor
\State Recover $\Wm$ from~\eqref{eq:milac_param}.
\State \textbf{Return} $(\Wm, \pv)$ and $\mathrm{EE}^\star = \lambda^{(n)}$.
\end{algorithmic}
\end{algorithm}

\begin{proof}
Define $\Dm=\operatorname{diag}(\pv)$.

\emph{Convexity:} 
Vectorizing $\Ym$ in the objective $f(\Ym)$ yields
\begin{equation}
\begin{aligned}
f\left(\Ym\right)
=&
\operatorname{vec}\left(\Ym\right)^\HH
\left(\Dm^{\TT}\otimes \Qm\right)
\operatorname{vec}\left(\Ym\right)\\
&-
2\Re\left\{
\operatorname{vec}\left(\Lm^\HH \Dm^{1/2}\right)^\HH
\operatorname{vec}\left(\Ym\right)
\right\},    
\end{aligned}
\end{equation}
up to an irrelevant constant. Hence the Hessian of $f(\Ym)$ with respect to $\operatorname{vec}(\Ym)$ is $\nabla^2 f(\Ym)=2(\Dm^{\TT}\otimes\Qm)\succeq\zerov$, because the Kronecker product of two positive semidefinite matrices is positive semidefinite. Therefore, $f(\Ym)$ is convex in $\Ym$.

\emph{Lipschitz gradient:} For any $\Ym_1,\Ym_2$,
\begin{align}
\left\|\nabla_{\Ym} f\left(\Ym_1\right)-\nabla_{\Ym} f\left(\Ym_2\right)\right\|_{\mathrm F}
&=
2\left\|\Qm\left(\Ym_1-\Ym_2\right)\Dm\right\|_{\mathrm F} \\
&\le
2\left\|\Qm\right\|_2\,\left\|\Dm\right\|_2\,\left\|\Ym_1-\Ym_2\right\|_{\mathrm F},\notag
\end{align}
confirming $L_f = 2\|\Qm\|_2\|\Dm\|_2$ as stated in~\eqref{eq:Lf}.

\emph{Convergence:} The set $\mathcal{C}$ is closed, convex, and compact. The update~\eqref{eq:Ystep}, obtained by singular-value clipping, is the exact Euclidean projection onto $\mathcal{C}$. Euclidean projections onto closed convex sets are firmly nonexpansive~\cite{Bertsekas1999Nonlinear}. Since $f(\Ym)$ is convex with $L_f$-Lipschitz gradient and $\mathcal{C}$ is compact and convex, the sequence converges to a global minimizer of $f(\Ym)$ over $\mathcal{C}$ for any $\tau\in(0,1/L_f]$~\cite{Bertsekas1999Nonlinear}.
\end{proof}

\subsection{Overall Algorithm, Convergence, and Complexity}\label{sec:algo_summary}

We summarize the complete solution developed in Sections~\ref{sec:reformulation}--\ref{sec:wmmse}. 
The inner loop is detailed in Algorithm~\ref{algo:inner} while the complete \ac{EE} optimizer wrapping the Dinkelbach outer iteration around it is given in Algorithm~\ref{algo:ee_overall}.

\begin{remark}[Overall Convergence]\label{rem:convergence}
Algorithm~\ref{algo:ee_overall} features two nested loops whose convergence properties are as follows.

\emph{Inner loop (Algorithm~\ref{algo:inner}):}
The receiver and weight updates~\eqref{eq:u_update}--\eqref{eq:omega_update} are globally optimal for fixed beamformers, and the power update~\eqref{eq:p_update} solves its separable convex subproblem in closed form. Both provide monotone descent on the \ac{WMMSE} objective~\eqref{eq:wmmse_problem}. The $\Ym$-update (Proposition~\ref{prop:Y_update_convergence}) solves its convex subproblem over $\mathcal{C}$ to global optimality via \ac{PGD} while joint power-budget feasibility is maintained by the $\pv$-update at every cycle. Therefore, the alternating minimization converges to a stationary point of~\eqref{eq:wmmse_problem}~\cite{Bertsekas1999Nonlinear}.

\emph{Outer loop (Dinkelbach).}
The Dinkelbach parameter $\lambda^{(n)}$ equals the \ac{EE} achieved by the inner-loop solution at iteration $n$. Because Algorithm~\ref{algo:inner} returns a point whose subtractive-form value is nonnegative (the previous iterate is always feasible), the sequence $\{\lambda^{(n)}\}$ is monotonically nondecreasing and upper-bounded by the finite global \ac{EE} optimum. Therefore, $\{\lambda^{(n)}\}$ converges~\cite{Dinkelbach1967}. The limit point is guaranteed to be a stationary point of the original fractional \ac{EE} problem~\eqref{eq:problem_reduced}.
\end{remark}

\begin{remark}[Complexity Comparison]\label{rem:complexity}
The overall per-channel complexity of Algorithm~\ref{algo:ee_overall} is $\mathcal{O}(NK^2 + T_D T_{M} T_{\mathrm{pgd}}K^3)$, where $T_D$, $T_{M}$, and $T_{\mathrm{pgd}}$ denote the Dinkelbach, \ac{WMMSE} alternation, and \ac{PGD} iterations, respectively. The leading $NK^2$ term accounts for the one-time Gram-matrix computation $\bar{\Hm}=\Hm\Hm^\HH$. Because the reduced parameterization (Proposition~\ref{prop:reduced_param}) replaces the $N\times K$ beamformer $\Wm$ with the $K\times K$ variable $\Ym$, every iterative operation, the $\Ym$-update \ac{SVD}~\eqref{eq:Ystep} at $\mathcal{O}(K^3)$, the receiver, weight, and power updates each at $\mathcal{O}(K^2)$, scales exclusively with $K$ and is independent of the antenna number $N$. Without exploiting this reduction, we would have to solve the problem in the original $N\times K$ beamformer space, incurring $\mathcal{O}(N^3)$ or higher per-iteration cost (if \ac{SDP} is invoked). This complexity reduction is made possible by the structural property of \ac{MiLAC}-aided beamforming identified in Remark~\ref{rem:rowspace_advantage}, which \emph{does not} generally hold for digital or hybrid beamforming benchmarks, so the same reduction is not exact for their quantization-aware problems, which are therefore solved in the full $N$-dimensional space. Concretely, the quantization-aware benchmark solvers of Section~\ref{sec:setup} share the Dinkelbach-\ac{WMMSE} structure of Algorithm~\ref{algo:ee_overall} but update their beamformers in that full space: the digital update solves the regularized $N\times N$ linear systems of the classical \ac{WMMSE} beamformer step~\cite{Shi2011Iteratively}, and the hybrid analog-precoder update is a \ac{PGD} step on the constant-modulus manifold~\cite{Tranter2017Fast}, whose step size is set by the spectral norm of an $N\times N$ weighted channel Gram matrix, the counterpart of $L_f$ in Proposition~\ref{prop:Y_update_convergence}. Both therefore cost $\mathcal{O}(N^3)$ per iteration, in addition to an $\mathcal{O}(N_{\mathrm{RF}}^3)$ baseband update for the hybrids.
\end{remark}

\begin{remark}[Initialization and Sensitivity]\label{rem:init_sens}
Algorithms~\ref{algo:inner} and~\ref{algo:ee_overall} start from $\pv^{(0)}=\frac{P_{\max}}{\alpha_q K}\onev$ and $\Ym^{(0)}=\bar{\Hm}^{1/2}/\|\bar{\Hm}^{1/2}\|_2$, which satisfy~\eqref{eq:problem_reduced_a}--\eqref{eq:problem_reduced_b} by construction. Since~\eqref{eq:problem_reduced} is nonconvex, different initial pairs can converge to different stationary points. In our experiments this dependence is mild: the \ac{EE} reached from this pair stays close on average to the best value found by random feasible restarts, and the \ac{EE} ordering of the four architectures reported in Section~\ref{sec:results} is the same whether Algorithm~\ref{algo:ee_overall} starts from this pair or from the best restart.
\end{remark}

\section{\ac{SE}-\ac{EE} Tradeoff Characterization}\label{sec:continuation}

\subsection{Multi-Objective Formulation and Weighted-Sum Approach}\label{sec:se_ee_range}

The \ac{EE}-maximizing operating point from Section~\ref{sec:lc_solver} does not reveal how much \ac{SE} is sacrificed to obtain the \ac{EE} gain. Characterizing this tradeoff is fundamentally a \ac{MOO} problem~\cite{Ehrgott2005Multicriteria}, in which \ac{SE} and \ac{EE} are jointly maximized over the same \ac{MiLAC}-aided beamforming feasible set, expressed as
\begin{subequations}\label{eq:moo}
\begin{alignat}{2}
    &\max_{\Ym,\pv}\quad &&
    \bigl[R\left(\Ym,\pv\right),\;\operatorname{EE}\left(\Ym,\pv\right)\bigr] \label{eq:moo_obj}\\
    &\text{s.t.}\quad &&
    \|\Ym\|_2 \le 1, \\
    & &&
    \alpha_q\onev^{\mathsf T}\pv \le P_{\max}.
\end{alignat}
\end{subequations}
In the moderate-to-high power regime, the two objectives conflict: higher \ac{SE} demands more transmit power, which reduces \ac{EE}~\cite{Amin2016SEEE,Zhou2022RSMA}.

A standard approach to trace the Pareto boundary of~\eqref{eq:moo} is the weighted-sum method~\cite{Ehrgott2005Multicriteria,Amin2016SEEE,Zhou2022RSMA}. By Proposition~3.9 of~\cite{Ehrgott2005Multicriteria}, the optimal solution of the corresponding weighted-sum problem is weakly Pareto optimal for the original \ac{MOO} problem. To make the two objectives commensurable, each is normalized by its own endpoint value, yielding
\begin{subequations}\label{eq:tradeoff_utility}
\begin{alignat}{2}
    &\max_{\Ym,\pv}\quad &&
    \eta \frac{R\left(\Ym,\pv\right)}{R_{\mathrm{ref}}}
    +
    (1-\eta)
    \frac{\operatorname{EE}\left(\Ym,\pv\right)}{\operatorname{EE}_{\mathrm{ref}}} \\
    &\text{s.t.}\quad &&
    \|\Ym\|_2 \le 1, \\
    & &&
    \alpha_q\onev^{\mathsf T}\pv \le P_{\max},
\end{alignat}
\end{subequations}
where $\eta \in [0,1]$ is the weight, $R_{\mathrm{ref}} = R_{\eta=1}$ is the \ac{SE}-only endpoint, and $\operatorname{EE}_{\mathrm{ref}} = \operatorname{EE}_{\eta=0}$ is the \ac{EE}-only endpoint of the architecture. Setting $\eta=0$ recovers the \ac{EE} problem of Section~\ref{sec:lc_solver}, $\eta=1$ gives the \ac{SE} problem, and $0<\eta<1$ generates interior tradeoff points. By sweeping $\eta$ over a grid, one obtains a set of weakly Pareto optimal points whose image in the raw $(R,\operatorname{EE})$ plane describes the supported portion of the tradeoff boundary. Non-supported (concave-region) Pareto points, if any, are not recoverable by the weighted-sum method~\cite{Ehrgott2005Multicriteria}.

\subsection{Auxiliary-Variable Reformulation and \ac{SCA} Solver}\label{sec:aux_reform}

The bilinear couplings in $(\Ym,\pv)$ prevent a jointly convex surrogate. Following the proof of Proposition~\ref{prop:reduced_param}, we let $\Wm=\Hm^\HH\Xm$, where $\Xm \in \mathbb{C}^{K \times K}$. This parameterization is lossless by Lemma~\ref{lem:range_restrict} and linearizes the coupling matrix to 
\begin{equation}
\Cm=\bar{\Hm}\Xm,	
\end{equation}
rendering $e_k$ and the feasibility \ac{LMI} $\Xm^\HH\bar{\Hm}\Xm\preceq\operatorname{diag}(\pv)$ convex in $(\Xm,\pv)$, while the transmit power $\alpha_q\onev^{\mathsf T}\pv$ is linear in $\pv$. We then introduce auxiliary variables to decompose the remaining nonconvexity into two scalar constraints, each admitting a valid minorizer.

Introducing auxiliary variables $r,t,x,y\ge 0$, \eqref{eq:tradeoff_utility} is equivalently written as
\begin{subequations}\label{eq:tradeoff_lifted}
\begin{alignat}{2}
    &\max_{\Xm,\pv,r,t,x,y}\quad &&
    \eta \frac{r}{R_{\mathrm{ref}}} + (1-\eta)\frac{t}{\operatorname{EE}_{\mathrm{ref}}}
    \label{eq:tradeoff_lifted_obj}\\
    &\text{s.t.}\quad &&
    r \le R(\Xm), \label{eq:tradeoff_r}\\
    & &&
    t \le \frac{B\,x^2}{y}, \label{eq:tradeoff_t}\\
    & &&
    x^2 \le r, \label{eq:tradeoff_x}\\
    & &&
    \Xm^\HH\bar{\Hm}\Xm \preceq \operatorname{diag}(\pv), \label{eq:tradeoff_y}\\
    & &&
    \frac{\alpha_q\onev^{\mathsf T}\pv}{\eta_{\mathrm{PA}}}+P_{\mathrm{st}}^{\mathrm{milac}} \le y,\; \alpha_q\onev^{\mathsf T}\pv \le P_{\max}.\label{eq:tradeoff_feas}
\end{alignat}
\end{subequations}
At any optimum, constraints~\eqref{eq:tradeoff_r}--\eqref{eq:tradeoff_y} are active, recovering~\eqref{eq:tradeoff_utility}. The only nonconvex elements are the rate constraint~\eqref{eq:tradeoff_r} and the fractional constraint~\eqref{eq:tradeoff_t}, each of which admits a valid minorizer.

\subsubsection{Rate minorizer}
Applying the \ac{WMMSE} identity~\eqref{eq:wmmse_identity} at any $(\uv^{(n)},\omegav^{(n)})$ with $\omega_k^{(n)}>0$ gives
\begin{align}\label{eq:rate_minorizer}
R(\Xm) &\ge \tilde{R}\bigl(\Xm\bigr)\\
    & = \frac{1}{\ln 2}\sum_{k=1}^{K}\Bigl[\ln \omega_k^{(n)} + 1 - \omega_k^{(n)} e_k\bigl(\Xm\bigr)\Bigr],	\notag
\end{align}
with equality at the optimal values~\eqref{eq:u_update}--\eqref{eq:omega_update}. Because $e_k(\Xm)$ is convex quadratic in $\Xm$, $\tilde{R}(\Xm)$ is concave.

\subsubsection{Fractional minorizer}
The map $(x,y)\mapsto x^2/y$ is jointly convex, so its first-order expansion is a global lower bound,
\begin{equation}\label{eq:xy_minorizer}
    \frac{x^2}{y} \ge \phi^{(n)}(x,y) = \frac{2 x^{(n)}}{y^{(n)}}x - \biggl(\frac{x^{(n)}}{y^{(n)}}\biggr)^{\!2} y,
\end{equation}
with equality at $(x,y)=(x^{(n)},y^{(n)})$.

\looseness=-1 \subsubsection{Per-weight \ac{SCA} subproblem}
Replacing the right-hand side of~\eqref{eq:tradeoff_r} by $\tilde{R}$ and that of~\eqref{eq:tradeoff_t} by $B\phi^{(n)}$ yields
\begin{subequations}\label{eq:tradeoff_sca_sub}
\begin{alignat}{2}
    &\max_{\Xm,\pv,r,t,x,y}\quad &&
    \eta \frac{r}{R_{\mathrm{ref}}} + (1-\eta)\frac{t}{\operatorname{EE}_{\mathrm{ref}}} \\
    &\text{s.t.}\quad &&
    r \le \tilde{R}\bigl(\Xm\bigr), \label{eq:sca_r}\\
    & &&
    t \le B\,\phi^{(n)}(x,y), \label{eq:sca_t}\\
    & &&
    x^2 \le r, \label{eq:sca_x}\\
    & &&
    \Xm^\HH\bar{\Hm}\Xm \preceq \operatorname{diag}(\pv), \label{eq:sca_y}\\
    & &&
    \frac{\alpha_q\onev^{\mathsf T}\pv}{\eta_{\mathrm{PA}}}+P_{\mathrm{st}}^{\mathrm{milac}} \le y,\; \alpha_q\onev^{\mathsf T}\pv \le P_{\max}.
\end{alignat}
\end{subequations}
All constraints of~\eqref{eq:tradeoff_sca_sub} are convex and the objective is linear, so the subproblem is a convex program.

\subsubsection{Overall solver and convergence}
For each fixed $\eta$, the per-weight solver maintains an incumbent $(\Xm^{(n)},\pv^{(n)},r^{(n)},t^{(n)},x^{(n)},y^{(n)})$ and cycles through three updates: \emph{(B1)} update $(\uv^{(n)},\omegav^{(n)})$ to the \ac{WMMSE}-optimal values~\eqref{eq:u_update}--\eqref{eq:omega_update} at $\Xm^{(n)}$. The rate minorizer~\eqref{eq:rate_minorizer} becomes tight at $\Xm^{(n)}$, giving $\tilde{R}(\Xm^{(n)})=R(\Xm^{(n)})$, \emph{(B2)} take the incumbent $(x,y)$ as the expansion point $(x^{(n)},y^{(n)})$ in~\eqref{eq:xy_minorizer}, so that $\phi^{(n)}$ is tight at the current iterate, and \emph{(B3)} solve the convex program~\eqref{eq:tradeoff_sca_sub} and set the new iterate to its maximizer.

Both minorizers are tight at the current iterate, so the incumbent is always feasible for~\eqref{eq:tradeoff_sca_sub}. Because the subproblem inner-approximates~\eqref{eq:tradeoff_lifted}, feasibility of the original problem is preserved and the weighted-sum objective is monotonically non-decreasing. Boundedness then guarantees convergence of the objective sequence. The tightness and gradient consistency of the two minorizers further ensure that every limit point satisfies the \ac{KKT} conditions of~\eqref{eq:tradeoff_utility}~\cite[Appendix~A]{Zhou2022RSMA}. The effective beamformer is recovered as $\Wm=\Hm^\HH\Xm$.

\begin{algorithm}[t]
\caption{Weighted-Sum \ac{SE}-\ac{EE} Tradeoff Boundary Solver}
\label{algo:tradeoff}
\footnotesize
\begin{algorithmic}[1]
\State \textbf{Input:} channel $\Hm$, power budget $P_{\max}$, \ac{MiLAC} power parameters, weight grid $\eta \in \{0,0.05,\ldots,0.95,1\}$, tolerance $\epsilon_{\mathrm{out}}>0$.
\State Solve the \ac{EE}-only ($\eta=0$) problem by Algorithm~\ref{algo:ee_overall}, and the \ac{SE}-only ($\eta=1$) problem by its inner solver, Algorithm~\ref{algo:inner}, with the Dinkelbach parameter $\lambda=0$, so that the subtractive objective reduces to the rate.
\State Set references $R_{\mathrm{ref}} \leftarrow R_{\eta=1}$ and $\operatorname{EE}_{\mathrm{ref}} \leftarrow \operatorname{EE}_{\eta=0}$.
\For{each interior $\eta$ on the ordered grid}
\State Initialize $(\Xm^{(0)},\pv^{(0)})$ from a feasible point, set $r^{(0)}=R(\Xm^{(0)})$, $x^{(0)}=\sqrt{r^{(0)}}$, $y^{(0)}=P_{\mathrm{tot}}^{\mathrm{milac}}(\pv^{(0)})$, and $t^{(0)}=B(x^{(0)})^2/y^{(0)}$.
\Repeat
\State Update $(\uv^{(n)},\omegav^{(n)})$ via~\eqref{eq:u_update}--\eqref{eq:omega_update} at the current $\Xm^{(n)}$.
\State Set the fractional-minorizer reference $(x^{(n)},y^{(n)})$ to the incumbent values of $(x,y)$.
\State Solve the convex program~\eqref{eq:tradeoff_sca_sub} by an interior-point method to obtain the new $(\Xm,\pv,r,t,x,y)$.
\Until{$|U_\eta^{(n)} - U_\eta^{(n-1)}|/|U_\eta^{(n-1)}|\le \epsilon_{\mathrm{out}}$, where $U_\eta^{(n)}$ is the value of~\eqref{eq:tradeoff_lifted_obj} at the new iterate.}
\State Recover $\Wm_\eta = \Hm^\HH \Xm$ and record the achieved $(R_\eta, \operatorname{EE}_\eta)$.
\EndFor
\State \textbf{Return} $\{(R_\eta, \operatorname{EE}_\eta)\}_\eta$ and the \ac{SE}-\ac{EE} tradeoff boundary.
\end{algorithmic}
\end{algorithm}

\begin{remark}[Implementation Details]\label{rem:implementation}
Two practical measures improve the quality of the traced tradeoff boundary. First, each interior $\eta$ is warm-started from the converged solution at the nearest previously computed weight, which places the initial iterate near a good stationary branch. Second, after a full sweep of the $\eta$ grid, $R_{\mathrm{ref}}$ and $\operatorname{EE}_{\mathrm{ref}}$ are recalibrated from the best \ac{SE} and \ac{EE} values actually attained, and the sweep is repeated for a small number of passes until the boundary stabilizes. Both measures preserve the per-weight monotone-ascent and \ac{KKT} convergence properties established above.
\end{remark}

\begin{figure}[t]
    \centering
    \subfloat[Selected BS and scene layout of the fixed \texttt{asu\_campus\_3p5} deployment.]{
        \includegraphics[width=0.96\linewidth]{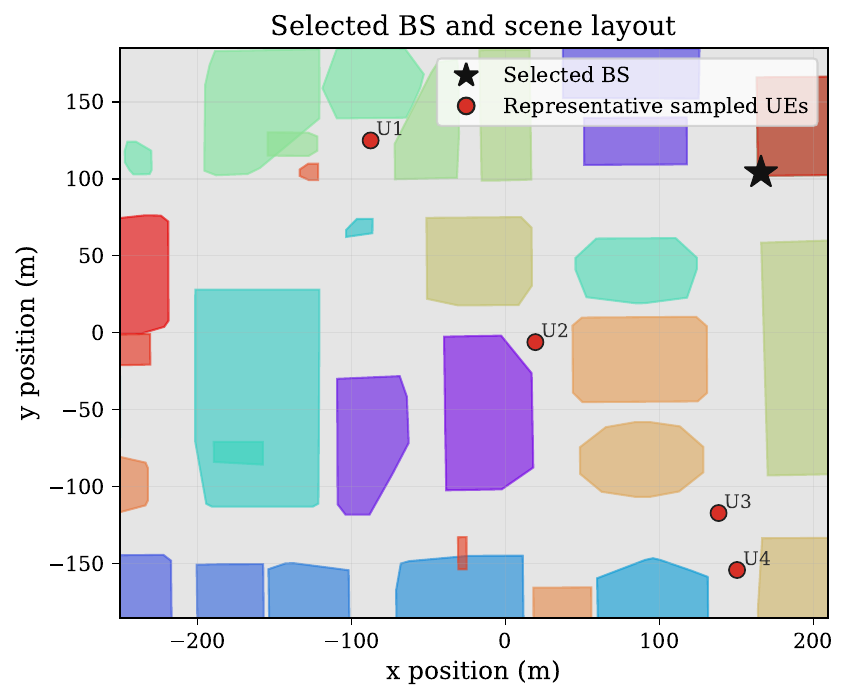}
    }


    \subfloat[Selected-TX path-loss coverage with the representative sampled users highlighted.]{
        \includegraphics[width=0.96\linewidth]{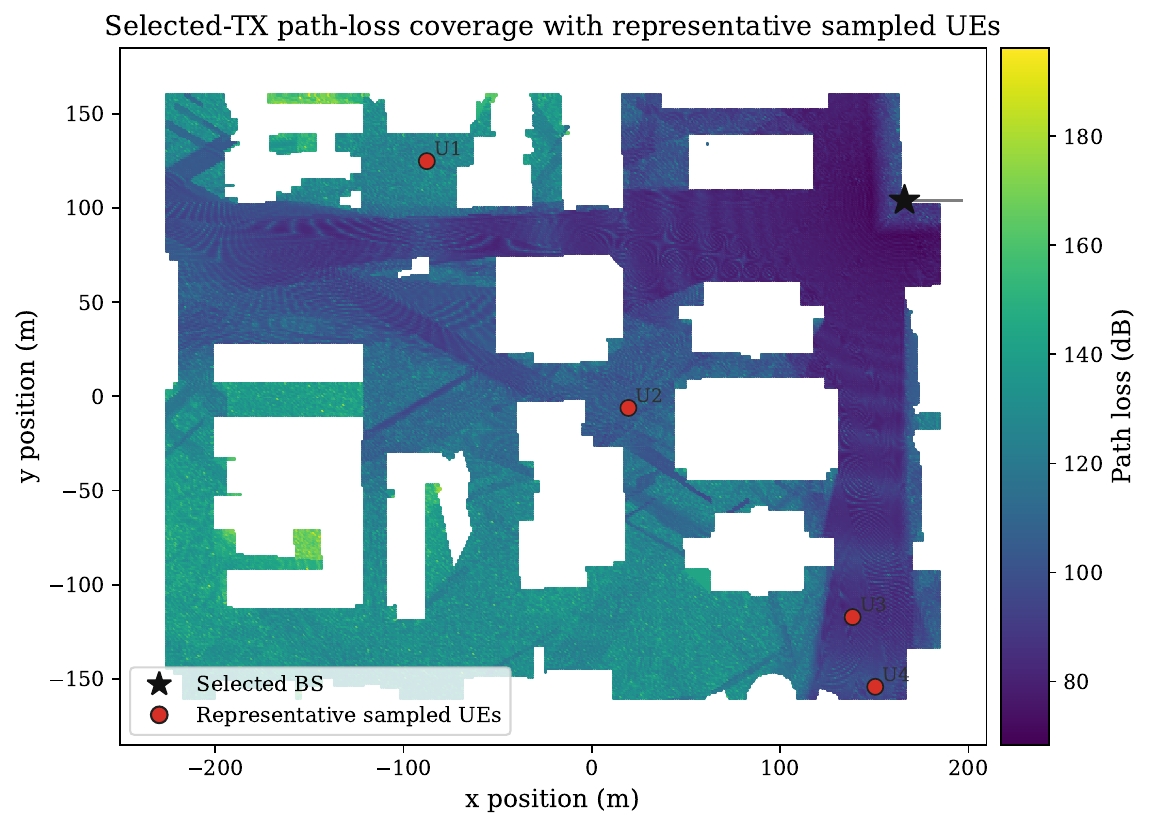}
    }
    \caption{Two-panel illustration of the adopted DeepMIMO simulation setting.}
    \label{fig:deepmimo_scene}
\end{figure}

\looseness=-1 \section{Numerical Results}\label{sec:results}
\subsection{Simulation Setup and Benchmarks}\label{sec:setup}
All results are obtained under a fixed DeepMIMO v4 deployment (\texttt{asu\_campus\_3p5})~\cite{DeepMIMOv4Website} and the transmitter-side power model of Section~\ref{sec:system_model}, as visualized in Fig.~\ref{fig:deepmimo_scene}. The raw channels include \ac{LoS}/\ac{NLoS} structure, path-loss variation, and valid-path filtering. One center subcarrier (index~$256$ of~$512$ at $100$~MHz bandwidth) is selected as the representative narrowband channel. The scenario contains a single \ac{BS}. Each Monte Carlo trial draws $K$ users at random, without replacement, from the $85\,157$-user valid pool under a fixed seed, and the same draws are shared by all architectures. The channel gains are rescaled so that the noise variance is one, which preserves all signal-to-noise ratios.

\begin{table}[t]
\caption{Default parameters.}
\label{tab:defaults}
\centering
\footnotesize
\begin{tabular}{l c}
\toprule
Parameter & Value \\
\midrule
Carrier frequency & $3.5$~GHz \\
Number of antennas $N$ & $64$ \\
Number of users $K$ & $4$ \\
Hybrid \ac{RF} chains $N_{\mathrm{RF}}$ & $4$ \\
Transmit power budget $P_{\max}$ & $30~\mathrm{dBm}$ \\
Bandwidth/sampling rate $B=f_s$ & $100~\mathrm{MHz}$~\cite{Choi2022Energy} \\
Noise figure & $5~\mathrm{dB}$~\cite{Choi2022Energy} \\
\ac{DAC} bits $b_{\mathrm{DAC}}$ & $4$ \\
\ac{PA} efficiency $\eta_{\mathrm{PA}}$ & $0.27$~\cite{Choi2022Energy,Ribeiro2018Energy} \\
\ac{RF} chain-related $P_{\mathrm{LP}}, P_{\mathrm{M}}, P_{\mathrm{H}}$ & $14,\,0.3,\,3~\mathrm{mW}$~\cite{Choi2022Energy} \\
Local oscillator $P_{\mathrm{LO}}$ & $22.5~\mathrm{mW}$~\cite{Choi2022Energy} \\
Phase shifter power $P_{\mathrm{PS}}$ & $21.6~\mathrm{mW}$~\cite{Ribeiro2018Energy} \\
Per-admittance drive-circuit power $P_{\mathrm{drv}}$ & $35~\mu\mathrm{W}$~\cite{Wang2024RISPower} \\
\bottomrule
\end{tabular}
\end{table}

Unless otherwise stated, the remaining default parameters are listed in Table~\ref{tab:defaults}. Direct \ac{MiLAC} hardware measurements are not yet available, so $P_{\mathrm{drv}}$ is taken from the \ac{RIS} drive-circuit measurements of~\cite{Wang2024RISPower}. There, the drive-circuit power scales with the number of control signals divided by $N_g$, the number of elements sharing one signal. This grouping factor, not the raw wattage, is what separates the reported per-element values by orders of magnitude. Every admittance of a \ac{MiLAC} is individually tunable, which fixes $N_g=1$. At $N_g=1$, \cite{Wang2024RISPower} measures an 8-bit shift register (SN74LV595A) drawing $I_{\mathrm{cc}}=20~\mu$A at $V_{\mathrm{cc}}=3.3$~V, that is $0.07$~mW for eight control signals, or $8.75~\mu$W each. Its much larger per-element values come from a continuous bias, where a \ac{DAC} and an operational amplifier supply $430$~mW to a group of $N_g=32$ cells. That arrangement does not apply here: each \ac{MiLAC} admittance needs a dedicated bias to be individually controlled, and a discrete-state admittance, realized for instance as a switched-reactance bank, needs no continuous bias at all.

Resolving an admittance to $b_{\mathrm{adm}}$ bits costs $b_{\mathrm{adm}}$ control signals, so $P_{\mathrm{drv}}=b_{\mathrm{adm}}\times 8.75~\mu$W. How fine a resolution the continuous design of Section~\ref{sec:continuation} needs is open for \ac{MiLAC}, but has been studied for the closely related BD-RIS, whose impedance network obeys the same kind of realization condition. That condition is underdetermined in an \ac{FC} architecture, since the tunable components outnumber the constraints, and the extra degrees of freedom compensate for the coarse resolution of each component. Exploiting this, \cite{Nerini2023Discrete} finds that a single bit already attains approximately the continuous-value performance, for any number of elements. An \ac{FC} \ac{MiLAC} is over-parameterized identically, with $(N{+}K)(N{+}K{+}1)/2=2346$ admittances for the $2NK=512$ real parameters of the beamformer they realize. We nevertheless set $b_{\mathrm{adm}}=4$, four times that resolution, giving $P_{\mathrm{drv}}=35~\mu$W. A \ac{MiLAC} driver may differ from a \ac{RIS} driver, so Section~\ref{sec:ee_performance} sweeps $P_{\mathrm{drv}}$ over two orders of magnitude, and \ac{MiLAC}-aided beamforming keeps the highest \ac{EE} up to $0.56$~mW per admittance, $16$ times the adopted value. If a prototype exceeds that value, reducing the number of admittances helps as much as reducing the power per admittance, since a stem-connected network needs $\mathcal{O}(NK)$ drive circuits rather than $\mathcal{O}((N{+}K)^2)$ with no significant sum-rate loss~\cite{Zhang2026StemBF}, which is $4.5\times$ fewer here and more than an order of magnitude fewer for large arrays. A quantitative study of this design is left to future work.

Four architectures are compared: digital beamforming ($N$ \ac{RF} chains), \ac{FC} hybrid beamforming (hybrid-\ac{FC}: $N_{\mathrm{RF}}$ chains plus $N N_{\mathrm{RF}}$ \acp{PS}), \ac{SC} hybrid beamforming (hybrid-\ac{SC}: $N_{\mathrm{RF}}$ chains plus $N$ \acp{PS} with one per antenna), and \ac{MiLAC}-aided beamforming ($K$ chains plus $(N{+}K)(N{+}K{+}1)/2$ admittances). In the hybrid-\ac{SC}, each \ac{RF} chain is connected exclusively to a dedicated sub-array of $\lfloor N/N_{\mathrm{RF}} \rfloor$ (or $\lceil N/N_{\mathrm{RF}} \rceil$) antennas, so the analog beamformer $\Fm_{\mathrm{RF}}$ is block-diagonal. All four are optimized under the same \ac{DAC}-\ac{AQNM} and post-quantized transmit power budget using quantization-aware \ac{WMMSE} updates tailored to each architecture. 

For fairness, all four solvers share a common outer structure: a Dinkelbach loop converts the fractional \ac{EE} objective into a subtractive-form inner problem, and the inner problem is solved by \ac{WMMSE}-based alternating optimization with MMSE receiver, MSE weight, and beamformer updates. The \ac{AQNM} quantization parameters $(\alpha_q,\beta_q)$ enter all architectures identically, and the Dinkelbach parameter $\lambda$ is updated by the same rule~\eqref{dinkel_para}. What differs across architectures is the beamformer update and the distortion structure it induces. For digital beamforming, the beamformer $\Wm \in \mathbb{C}^{N \times K}$ is updated in closed form via a regularized linear system with an $N \times N$ inverse~\cite{Shi2011Iteratively}, and the \ac{AQNM} distortion covariance is diagonal in the antenna domain, depending on the per-antenna row norms of $\Wm$. For hybrid beamforming (\ac{FC} and \ac{SC}), the precoder factorizes as $\Wm = \Fm_{\mathrm{RF}}\Fm_{\mathrm{BB}}$, the baseband precoder $\Fm_{\mathrm{BB}} \in \mathbb{C}^{N_{\mathrm{RF}} \times K}$ admits a closed-form update of the same \ac{WMMSE}-derived regularized form as the digital one, operating in the effective channel $\Gm = \Hm\Fm_{\mathrm{RF}}$, while the analog precoder $\Fm_{\mathrm{RF}}$ is updated via projected gradient descent onto the constant-modulus manifold~\cite{Tranter2017Fast} (element-wise for hybrid-\ac{FC}, block-wise for hybrid-\ac{SC}). The distortion covariance is diagonal in the \ac{RF}-chain domain, propagated to the antenna domain through $\Fm_{\mathrm{RF}}$. 

\begin{table}[t]
\caption{Per-iteration complexity and measured runtime of the four \ac{EE} solvers at $K=N_{\mathrm{RF}}=4$ and $P_{\max}=30$~dBm.}
\label{tab:runtime}
\centering
\footnotesize
\setlength{\tabcolsep}{3.5pt}
\begin{tabular}{l c c c c c}
\toprule
Architecture & Per-iteration cost & $T_D$ & $T_M$ & [s] $N{=}64$ & [s] $N{=}256$ \\
\midrule
Digital & $\mathcal{O}(N^3)$ & 4.1 & 13 & 0.35 & 9.06 \\
MiLAC-aided & $\mathcal{O}(K^3)$ & 5.0 & 8 & 0.09 & 0.09 \\
Hybrid-FC & $\mathcal{O}(N^3)$ & 3.5 & 34 & 0.36 & 5.69 \\
Hybrid-SC & $\mathcal{O}(N^3)$ & 9.7 & 33 & 0.54 & 5.85 \\
\bottomrule
\end{tabular}

\smallskip

\parbox{0.98\linewidth}{\scriptsize $T_D$: Dinkelbach rounds and $T_M$: WMMSE iterations of the returned round, both at $N=64$; runtime: mean time to solve one channel realization to convergence, over 20 DeepMIMO realizations at $N=64$ and 10 at $N=256$. Platform: Apple M5.}

\end{table}

Table~\ref{tab:runtime} compares the computational cost of the four solvers under identical channels and stopping rules, timing one channel realization on a single CPU core. Hybrid-\ac{SC} needs the most Dinkelbach rounds and \ac{MiLAC} the fewest \ac{WMMSE} iterations, and the \ac{MiLAC} solver is already the fastest at $N=64$. Raising $N$ to $256$ leaves its runtime nearly unchanged, since its iterations run in the $K\times K$ space after the one-time $\mathcal{O}(NK^2)$ reduction, while the digital and hybrid runtimes grow by more than an order of magnitude (Remark~\ref{rem:complexity}). Figure~\ref{fig:convergence} shows the corresponding behavior on a representative realization: the \ac{EE} rises monotonically and the relative Dinkelbach gap falls below the $3\times 10^{-3}$ stopping tolerance within three (hybrid-\ac{FC}) to eleven (hybrid-\ac{SC}) outer rounds, five for \ac{MiLAC}. The reduction is specific to \ac{MiLAC}-aided beamforming: its stream-side \acp{DAC} produce the distortion covariance $\alpha_q\beta_q\Wm\Wm^\HH$, so the problem depends on $\Wm$ only through the reduced variables of Proposition~\ref{prop:reduced_param}, whereas per-antenna and per-chain quantization make the digital and hybrid distortion depend on the individual antenna and chain signal powers rather than on $\Hm\Wm$ alone, so the same reduction is not exact for them, and the hybrids additionally constrain every entry of the analog precoder (Remark~\ref{rem:rowspace_advantage}).

\begin{figure}[t]
\centering
\includegraphics[width=0.99\linewidth]{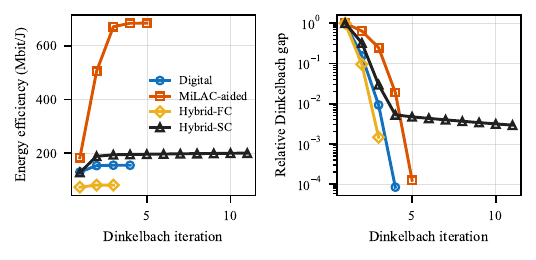}
\vspace{-0.25cm}
\caption{\ac{EE} and relative Dinkelbach gap versus outer iteration for the four solvers on a representative realization at the default setting.}
\label{fig:convergence}
\end{figure}

\begin{figure}[t]
    \centering
    \resizebox{0.99\linewidth}{!}{
%
\definecolor{mycolor1}{rgb}{0.06600,0.44300,0.74500}%
\definecolor{mycolor2}{rgb}{0.86600,0.32900,0.00000}%
\definecolor{mycolor3}{rgb}{0.92900,0.69400,0.12500}%
\definecolor{mycolor4}{rgb}{0.12941,0.12941,0.12941}%

\begin{tikzpicture}

\begin{axis}[%
width=0.5\linewidth,
height=0.38\linewidth,
at={(0\linewidth,0\linewidth)},
scale only axis,
xmin=10.0,
xmax=40.0,
xlabel style={font=\footnotesize\color{mycolor4}},
xlabel={Transmit power budget (dBm)},
ymin=0,
ymax=1000,
ylabel style={font=\footnotesize\color{mycolor4}},
ylabel={Energy efficiency (Mbit/J)},
axis background/.style={fill=white},
xmajorgrids,
ymajorgrids,
legend style={at={(0.6,0.80)}, font=\footnotesize, legend cell align=left, align=left},
tick label style={font=\footnotesize},
xtick={10,15,20,25,30,35,40}
]
\addplot [color=mycolor1, line width=1.5pt, mark=o, mark options={solid, mycolor1}]
  table[row sep=crcr]{%
10	66.5976795383314\\
15	111.976921897799\\
20	159.048832185905\\
25	181.63669260697\\
30	183.598742383248\\
35	183.668140818694\\
40	183.669763707392\\
};
\addlegendentry{Digital}

\addplot [color=mycolor2, line width=1.5pt, mark=square, mark options={solid, mycolor2}]
  table[row sep=crcr]{%
10	560.721803890553\\
15	754.206203596186\\
20	784.400066992882\\
25	785.44170588257\\
30	785.467633566816\\
35	785.469027267048\\
40	785.469890664951\\
};
\addlegendentry{MiLAC-aided}

\addplot [color=mycolor3, line width=1.5pt, mark=diamond, mark options={solid, mycolor3}]
  table[row sep=crcr]{%
10	28.1686293873176\\
15	46.9517902826134\\
20	71.0512978574347\\
25	87.7721860054469\\
30	92.2159955491489\\
35	93.7159827967518\\
40	94.3257952105108\\
};
\addlegendentry{Hybrid-FC}

\addplot [color=mycolor4, line width=1.5pt, mark=triangle, mark options={solid, mycolor4}]
  table[row sep=crcr]{%
10	89.5547609285692\\
15	148.057295233677\\
20	200.710082562118\\
25	216.83991258139\\
30	217.492873890206\\
35	217.839157611466\\
40	217.839838292767\\
};
\addlegendentry{Hybrid-SC}

\end{axis}

\begin{axis}[%
width=0.5\linewidth,
height=0.38\linewidth,
at={(0.62\linewidth,0\linewidth)},
scale only axis,
xmin=10.0,
xmax=40.0,
xlabel style={font=\footnotesize\color{mycolor4}},
xlabel={Transmit power budget (dBm)},
ymin=1,
ymax=8,
ylabel style={font=\footnotesize\color{mycolor4}},
ylabel={Spectral efficiency (bit/s/Hz)},
axis background/.style={fill=white},
xmajorgrids,
ymajorgrids,
tick label style={font=\footnotesize},
xtick={10,15,20,25,30,35,40}
]
\addplot [color=mycolor1, line width=1.5pt, mark=o, mark options={solid, mycolor1}, forget plot]
  table[row sep=crcr]{%
10	1.71386267058257\\
15	2.97135397520455\\
20	4.62321613816641\\
25	6.35515315882293\\
30	6.7269983932975\\
35	6.76414803977574\\
40	6.76366274166614\\
};
\addplot [color=mycolor2, line width=1.5pt, mark=square, mark options={solid, mycolor2}, forget plot]
  table[row sep=crcr]{%
10	1.67525191950443\\
15	2.81441860141838\\
20	3.36902013396958\\
25	3.42842917385334\\
30	3.43866133925455\\
35	3.43880249195186\\
40	3.43898354382704\\
};
\addplot [color=mycolor3, line width=1.5pt, mark=diamond, mark options={solid, mycolor3}, forget plot]
  table[row sep=crcr]{%
10	1.61864184840562\\
15	2.73557159284547\\
20	4.31962786256377\\
25	5.96950305656575\\
30	6.9383181350814\\
35	7.195918708171\\
40	7.29062179538284\\
};
\addplot [color=mycolor4, line width=1.5pt, mark=triangle, mark options={solid, mycolor4}, forget plot]
  table[row sep=crcr]{%
10	1.43203170662998\\
15	2.47947840413809\\
20	3.87079655840394\\
25	4.84647843427372\\
30	4.95000365850067\\
35	4.97463088154658\\
40	4.97475008038713\\
};
\end{axis}

\begin{axis}[%
width=1.247\linewidth,
height=0.474\linewidth,
at={(-0.162\linewidth,-0.058\linewidth)},
scale only axis,
xmin=0,
xmax=1,
ymin=0,
ymax=1,
axis line style={draw=none},
ticks=none,
]
\end{axis}

\end{tikzpicture}%
}
    \vspace{-0.5cm}
    \caption{\ac{EE} and achieved \ac{SE} versus transmit-power budget across architectures.
    }
    \label{fig:power}
\end{figure}

\subsection{\ac{EE}/\ac{SE} Performance Under \ac{EE}-Oriented Designs}\label{sec:ee_performance}
Figure~\ref{fig:power} compares the architectures over a $10$--$40$~dBm transmit-power sweep. \ac{MiLAC}-aided beamforming achieves the highest \ac{EE} throughout.
The paired \ac{SE} panel shows that this comes at a moderate \ac{SE} cost: \ac{MiLAC}'s \ac{EE}-optimal point levels off near $3.4$~bit/s/Hz beyond $25$~dBm, while digital and hybrid-\ac{FC} reach $6.76$ and $7.29$~bit/s/Hz at $40$~dBm, respectively. The reason is that raising $P_{\max}$ increases the \ac{PA} supply power and thereby exposes each architecture's circuit-power floor. \ac{MiLAC} benefits most, since its active front end remains stream-scaled and its passive network is largely insensitive to the transmit power.

\ac{MiLAC}'s \ac{SE} shortfall originates from the feasibility constraint $\|\Ym\|_2 \le 1$ (Proposition~\ref{prop:reduced_param}), which caps the power steered into any singular direction of $\bar{\Hm}^{1/2}$. Digital beamforming, constrained only by the aggregate budget $\|\Wm\|_{\mathrm{F}}^2 \le P_{\max}/\alpha_q$, can load more power onto weak eigenmodes for higher multiplexing gain, at the cost of a heavier active front end. The \ac{EE} framework exploits this asymmetry, trading a moderate \ac{SE} shortfall for a much lighter power footprint. The gap in Fig.~\ref{fig:power} thus mixes the structural cost of $\|\Ym\|_2\le 1$ with the transmit-power backoff selected by the \ac{EE} objective. Section~\ref{sec:tradeoff_discussion} separates the two at the maximum-\ac{SE} endpoints of the tradeoff boundaries, where the backoff is absent.

\begin{table}[t]
\caption{Power consumption breakdown across architectures (W).}
\label{tab:power_components}
\centering
\footnotesize
\resizebox{\linewidth}{!}{%
\begin{tabular}{lcccccc}
\toprule
Architecture & RF{+}DAC & PA supply & Phase shifters & MiLAC static & Common & Total \\
\midrule
Digital & 2.514 & 1.344 & 0.000 & 0.000 & 0.022 & 3.880 \\
MiLAC-aided & 0.157 & 0.387 & 0.000 & 0.082 & 0.022 & 0.649 \\
Hybrid-FC & 0.157 & 2.072 & 5.530 & 0.000 & 0.022 & 7.782 \\
Hybrid-SC & 0.157 & 0.988 & 1.382 & 0.000 & 0.022 & 2.550 \\
\bottomrule
\end{tabular}
}

\end{table}

Table~\ref{tab:power_components} reports consumption at $P_{\max}=30$~dBm. Hybrid-\ac{FC} ($7.78$~W) is dominated by its $N N_{\mathrm{RF}} P_{\mathrm{PS}} \approx 5.53$~W \ac{PS} power, which inverts the expected hybrid-versus-digital ordering~\cite{Ribeiro2018Energy}. Hybrid-\ac{SC} ($2.55$~W) reduces this term to $N P_{\mathrm{PS}} \approx 1.38$~W and thereby falls below digital ($3.88$~W). \ac{MiLAC} ($0.65$~W) is the smallest: it uses only $K$ active chains and a small per-admittance consumption ($\approx 0.08$~W), and it is the only design whose $N$-dependent cost involves no per-chain \ac{RF}/\ac{DAC} contribution, scaling instead as $(N+K)^2$ through a passive, \ac{FC} admittance network.

\begin{figure}[t]
\centering
\resizebox{0.99\linewidth}{!}{
%
\definecolor{mycolor1}{rgb}{0.06600,0.44300,0.74500}%
\definecolor{mycolor2}{rgb}{0.86600,0.32900,0.00000}%
\definecolor{mycolor3}{rgb}{0.92900,0.69400,0.12500}%
\definecolor{mycolor4}{rgb}{0.12941,0.12941,0.12941}%

\begin{tikzpicture}

\begin{axis}[%
width=0.5\linewidth,
height=0.38\linewidth,
at={(0\linewidth,0\linewidth)},
scale only axis,
xmin=2.0,
xmax=16.0,
xlabel style={font=\footnotesize\color{mycolor4}},
xlabel={Number of users},
ymin=0,
ymax=1400,
ylabel style={font=\footnotesize\color{mycolor4}},
ylabel={Energy efficiency (Mbit/J)},
axis background/.style={fill=white},
xmajorgrids,
ymajorgrids,
legend style={at={(0.6,0.80)}, font=\footnotesize, legend cell align=left, align=left},
tick label style={font=\footnotesize},
xtick={2,4,6,8,10,12,16}
]
\addplot [color=mycolor1, line width=1.5pt, mark=o, mark options={solid, mycolor1}]
  table[row sep=crcr]{%
2	124.435204010376\\
4	203.265110990414\\
6	282.85880211328\\
8	355.506428047696\\
10	420.21440815292\\
12	476.194102116257\\
16	555.983315332433\\
};
\addlegendentry{Digital}

\addplot [color=mycolor2, line width=1.5pt, mark=square, mark options={solid, mycolor2}]
  table[row sep=crcr]{%
2	717.615688332551\\
4	871.001733138644\\
6	1035.64512079652\\
8	1128.98684665626\\
10	1172.32299939472\\
12	1192.43815974053\\
16	1175.74175500371\\
};
\addlegendentry{MiLAC-aided}

\addplot [color=mycolor3, line width=1.5pt, mark=diamond, mark options={solid, mycolor3}]
  table[row sep=crcr]{%
2	103.195738250953\\
4	102.004661408974\\
6	104.423866776242\\
8	104.948440045287\\
10	105.162988158751\\
12	103.644508221083\\
16	96.1481738966043\\
};
\addlegendentry{Hybrid-FC}

\addplot [color=mycolor4, line width=1.5pt, mark=triangle, mark options={solid, mycolor4}]
  table[row sep=crcr]{%
2	159.71462185917\\
4	229.808499395779\\
6	288.633604949456\\
8	344.334397140659\\
10	390.006710098842\\
12	434.844110387832\\
16	492.450913546167\\
};
\addlegendentry{Hybrid-SC}

\end{axis}

\begin{axis}[%
width=0.5\linewidth,
height=0.38\linewidth,
at={(0.62\linewidth,0\linewidth)},
scale only axis,
xmin=2.0,
xmax=16.0,
xlabel style={font=\footnotesize\color{mycolor4}},
xlabel={Number of users},
ymin=0,
ymax=28,
ylabel style={font=\footnotesize\color{mycolor4}},
ylabel={Spectral efficiency (bit/s/Hz)},
axis background/.style={fill=white},
xmajorgrids,
ymajorgrids,
tick label style={font=\footnotesize},
xtick={2,4,6,8,10,12,16}
]
\addplot [color=mycolor1, line width=1.5pt, mark=o, mark options={solid, mycolor1}, forget plot]
  table[row sep=crcr]{%
2	4.51168747669652\\
4	7.44908021536441\\
6	10.2530295510957\\
8	12.8109813968644\\
10	15.2080121975945\\
12	17.3003196453138\\
16	20.1947631952749\\
};
\addplot [color=mycolor2, line width=1.5pt, mark=square, mark options={solid, mycolor2}, forget plot]
  table[row sep=crcr]{%
2	2.17821501694411\\
4	3.81607514468575\\
6	5.62190742476523\\
8	7.4044227830883\\
10	8.99719665312743\\
12	10.6647601439923\\
16	13.2636023189821\\
};
\addplot [color=mycolor3, line width=1.5pt, mark=diamond, mark options={solid, mycolor3}, forget plot]
  table[row sep=crcr]{%
2	4.0237682931164\\
4	7.67896913169274\\
6	11.4857469519356\\
8	15.0792457758724\\
10	18.5293485407693\\
12	21.4292236361465\\
16	25.4533865552307\\
};
\addplot [color=mycolor4, line width=1.5pt, mark=triangle, mark options={solid, mycolor4}, forget plot]
  table[row sep=crcr]{%
2	3.43181997575113\\
4	5.38734549189107\\
6	7.03911343870401\\
8	8.79720837024634\\
10	10.4349319423377\\
12	12.2199519978546\\
16	14.8282343463695\\
};
\end{axis}

\begin{axis}[%
width=1.247\linewidth,
height=0.474\linewidth,
at={(-0.162\linewidth,-0.058\linewidth)},
scale only axis,
xmin=0,
xmax=1,
ymin=0,
ymax=1,
axis line style={draw=none},
ticks=none,
]
\end{axis}

\end{tikzpicture}%
}
\vspace{-0.5cm}
\caption{\ac{EE} and achieved \ac{SE} versus served user number across architectures.}
\label{fig:users}
\end{figure}

Figure~\ref{fig:users} traces \ac{EE} and \ac{SE} as the user number grows from $K{=}2$ to $K{=}16$. Each architecture's static-power denominator depends on $K$ differently. Digital's $N$-scaled \ac{RF}/\ac{DAC} cost is $K$-independent, so its \ac{EE} rises steadily with sum-\ac{SE}. Hybrid-\ac{FC}'s $K N P_{\mathrm{PS}}$ \ac{PS} cost grows nearly linearly with the sum \ac{SE}, so its \ac{EE} stays near $100$~Mbit/J. Hybrid-\ac{SC}'s $K$-independent \ac{PS} cost lets its \ac{EE} rise similarly to digital.
\ac{MiLAC}-aided beamforming's $(N{+}K)^{2}$ static cost is dominated by $N^{2}$ when $N \gg K$, so its \ac{EE} rises from $718$~Mbit/J at $K{=}2$ to $1192$~Mbit/J at $K{=}12$ and then decreases slightly to $1176$~Mbit/J at $K{=}16$, while retaining the widest \ac{EE} margin throughout. None of the four designs constrains per-user rates, and \ac{EE} maximization can leave users with weak channels at essentially zero rate, since a weak user adds little to the sum rate but still costs transmit power. This behavior is a property of the objective rather than of any architecture, so the comparison favors none of them. In practice, fairness across users is handled by the scheduler over many channel realizations, and extending~\eqref{eq:problem_full} with per-user rate constraints is left for future study.

\begin{figure}[t]
    \centering
    \resizebox{0.99\linewidth}{!}{
%
\definecolor{mycolor1}{rgb}{0.06600,0.44300,0.74500}%
\definecolor{mycolor2}{rgb}{0.86600,0.32900,0.00000}%
\definecolor{mycolor3}{rgb}{0.92900,0.69400,0.12500}%
\definecolor{mycolor4}{rgb}{0.12941,0.12941,0.12941}%

\begin{tikzpicture}

\begin{axis}[%
width=0.5\linewidth,
height=0.38\linewidth,
at={(0\linewidth,0\linewidth)},
scale only axis,
xmin=8.0,
xmax=256.0,
xlabel style={font=\footnotesize\color{mycolor4}},
xlabel={Number of antennas},
ymin=0,
ymax=1200,
ylabel style={font=\footnotesize\color{mycolor4}},
ylabel={Energy efficiency (Mbit/J)},
axis background/.style={fill=white},
xmajorgrids,
ymajorgrids,
legend style={at={(0.7,0.80)}, font=\footnotesize, legend cell align=left, align=left},
tick label style={font=\footnotesize},
xtick={8,16,32,64,128,192,256}
]
\addplot [color=mycolor1, line width=1.5pt, mark=o, mark options={solid, mycolor1}]
  table[row sep=crcr]{%
8	363.504145981817\\
16	334.148155241399\\
32	274.666108592057\\
64	211.975174043953\\
128	153.626019529345\\
192	124.46665564969\\
256	105.519904243038\\
};
\addlegendentry{Digital}

\addplot [color=mycolor2, line width=1.5pt, mark=square, mark options={solid, mycolor2}]
  table[row sep=crcr]{%
8	493.463370533447\\
16	677.979424060624\\
32	866.43141489546\\
64	952.962758613762\\
128	777.316332821544\\
192	577.227529821564\\
256	431.618157679456\\
};
\addlegendentry{MiLAC-aided}

\addplot [color=mycolor3, line width=1.5pt, mark=diamond, mark options={solid, mycolor3}]
  table[row sep=crcr]{%
8	203.476628292428\\
16	179.722266649949\\
32	139.367008848051\\
64	101.25622690047\\
128	68.2344774396549\\
192	52.7993023944212\\
256	43.7858138470886\\
};
\addlegendentry{Hybrid-FC}

\addplot [color=mycolor4, line width=1.5pt, mark=triangle, mark options={solid, mycolor4}]
  table[row sep=crcr]{%
8	341.356619747788\\
16	342.414182520769\\
32	303.880277206803\\
64	244.37679482562\\
128	177.008519557485\\
192	143.329149588306\\
256	123.500094242287\\
};
\addlegendentry{Hybrid-SC}

\end{axis}

\begin{axis}[%
width=0.5\linewidth,
height=0.38\linewidth,
at={(0.62\linewidth,0\linewidth)},
scale only axis,
xmin=8.0,
xmax=256.0,
xlabel style={font=\footnotesize\color{mycolor4}},
xlabel={Number of antennas},
ymin=0,
ymax=16,
ylabel style={font=\footnotesize\color{mycolor4}},
ylabel={Spectral efficiency (bit/s/Hz)},
axis background/.style={fill=white},
xmajorgrids,
ymajorgrids,
tick label style={font=\footnotesize},
xtick={8,16,32,64,128,192,256}
]
\addplot [color=mycolor1, line width=1.5pt, mark=o, mark options={solid, mycolor1}, forget plot]
  table[row sep=crcr]{%
8	2.34319347486565\\
16	3.60847025334451\\
32	5.25668703067879\\
64	7.51763328251665\\
128	10.4278051058358\\
192	12.3460209199599\\
256	13.7460987247074\\
};
\addplot [color=mycolor2, line width=1.5pt, mark=square, mark options={solid, mycolor2}, forget plot]
  table[row sep=crcr]{%
8	1.85944330190592\\
16	2.36991502461925\\
32	3.0118186277056\\
64	3.8872060511338\\
128	5.38526365638417\\
192	6.66707769623953\\
256	7.69522410324663\\
};
\addplot [color=mycolor3, line width=1.5pt, mark=diamond, mark options={solid, mycolor3}, forget plot]
  table[row sep=crcr]{%
8	2.95038350223213\\
16	4.20411924525496\\
32	5.61620937211797\\
64	7.43231643983414\\
128	9.21105594389037\\
192	10.3189488273528\\
256	11.0903048693634\\
};
\addplot [color=mycolor4, line width=1.5pt, mark=triangle, mark options={solid, mycolor4}, forget plot]
  table[row sep=crcr]{%
8	2.26949064156794\\
16	3.05374686892989\\
32	4.09203113373945\\
64	5.45402891974282\\
128	7.02041460726774\\
192	8.07711442413514\\
256	8.96955045512082\\
};
\end{axis}

\begin{axis}[%
width=1.247\linewidth,
height=0.474\linewidth,
at={(-0.162\linewidth,-0.058\linewidth)},
scale only axis,
xmin=0,
xmax=1,
ymin=0,
ymax=1,
axis line style={draw=none},
ticks=none,
]
\end{axis}

\end{tikzpicture}%
}
    \vspace{-0.5cm}
    \caption{\ac{EE} and achieved \ac{SE} versus antenna number across architectures.
    }
    \label{fig:antennas}
\end{figure}

Figure~\ref{fig:antennas} shows the antenna-number scaling. \ac{MiLAC}-aided beamforming leads across the entire sweep, already surpassing digital at $N{=}8$ ($493$ versus\ $364$~Mbit/J) because its $K{=}4$ stream-domain chains are cheaper than the $N{=}8$ digital front end. As $N$ grows, digital's \ac{EE} declines with its $N$-scaled \ac{RF}/\ac{DAC} cost, hybrid-\ac{FC} degrades faster under the $N N_{\mathrm{RF}}$ \ac{PS} burden, and hybrid-\ac{SC}'s $N$-scaled \ac{PS} cost lets it overtake digital by $N{=}16$. \ac{MiLAC} keeps only $K$ active chains regardless of $N$, paying a quadratic passive-network cost: its \ac{EE} peaks at $N{=}64$ ($953$~Mbit/J) and decreases to $432$~Mbit/J at $N{=}256$ as the $(N{+}K)^{2}$ cost begins to dominate. The \ac{SE} panel shows all architectures improve with $N$. Digital's larger \ac{SE} reflects that its heavier static overhead pushes its \ac{EE}-optimum toward high-$P_{\mathrm{tx}}$, high-\ac{SE} operation, whereas \ac{MiLAC}'s lighter static overhead favors lower-power operation. Under \ac{SE}-only optimization this gap would close as $N/K$ grows~\cite{Wu2026MiLAC}.

\begin{figure}[t]
    \centering
    \resizebox{0.99\linewidth}{!}{
%
\definecolor{mycolor1}{rgb}{0.06600,0.44300,0.74500}%
\definecolor{mycolor2}{rgb}{0.86600,0.32900,0.00000}%
\definecolor{mycolor3}{rgb}{0.92900,0.69400,0.12500}%
\definecolor{mycolor4}{rgb}{0.12941,0.12941,0.12941}%

\begin{tikzpicture}

\begin{axis}[%
width=0.5\linewidth,
height=0.38\linewidth,
at={(0\linewidth,0\linewidth)},
scale only axis,
xmin=1.0,
xmax=8.0,
xlabel style={font=\footnotesize\color{mycolor4}},
xlabel={DAC resolution (bits)},
ymin=0,
ymax=800,
ylabel style={font=\footnotesize\color{mycolor4}},
ylabel={Energy efficiency (Mbit/J)},
axis background/.style={fill=white},
xmajorgrids,
ymajorgrids,
legend style={at={(0.7,0.80)}, font=\footnotesize, legend cell align=left, align=left},
tick label style={font=\footnotesize},
xtick={1,2,3,4,5,6,7,8}
]
\addplot [color=mycolor1, line width=1.5pt, mark=o, mark options={solid, mycolor1}]
  table[row sep=crcr]{%
1	150.689696783504\\
2	173.23638198044\\
3	176.558454033412\\
4	173.141767079279\\
5	167.321385841094\\
6	160.202291863133\\
7	151.465655360665\\
8	140.07400571085\\
};
\addlegendentry{Digital}

\addplot [color=mycolor2, line width=1.5pt, mark=square, mark options={solid, mycolor2}]
  table[row sep=crcr]{%
1	318.014403313312\\
2	545.592846420734\\
3	677.87573022157\\
4	733.391506802986\\
5	743.692544230589\\
6	733.325699655305\\
7	712.722152684637\\
8	681.869784818336\\
};
\addlegendentry{MiLAC-aided}

\addplot [color=mycolor3, line width=1.5pt, mark=diamond, mark options={solid, mycolor3}]
  table[row sep=crcr]{%
1	40.1959524410044\\
2	64.2273790542775\\
3	80.112267332384\\
4	88.5710933895747\\
5	92.405893843305\\
6	94.7809697912504\\
7	95.1642208915006\\
8	95.3011334909607\\
};
\addlegendentry{Hybrid-FC}

\addplot [color=mycolor4, line width=1.5pt, mark=triangle, mark options={solid, mycolor4}]
  table[row sep=crcr]{%
1	104.769219622059\\
2	160.634628663858\\
3	191.037540263231\\
4	204.640474609517\\
5	209.136651850152\\
6	209.730827461281\\
7	208.856853323828\\
8	206.995733078578\\
};
\addlegendentry{Hybrid-SC}

\end{axis}

\begin{axis}[%
width=0.5\linewidth,
height=0.38\linewidth,
at={(0.62\linewidth,0\linewidth)},
scale only axis,
xmin=1.0,
xmax=8.0,
xlabel style={font=\footnotesize\color{mycolor4}},
xlabel={DAC resolution (bits)},
ymin=1,
ymax=8,
ylabel style={font=\footnotesize\color{mycolor4}},
ylabel={Spectral efficiency (bit/s/Hz)},
axis background/.style={fill=white},
xmajorgrids,
ymajorgrids,
tick label style={font=\footnotesize},
xtick={1,2,3,4,5,6,7,8}
]
\addplot [color=mycolor1, line width=1.5pt, mark=o, mark options={solid, mycolor1}, forget plot]
  table[row sep=crcr]{%
1	4.75482882570884\\
2	5.78845072523072\\
3	6.22789207499732\\
4	6.40616178754446\\
5	6.51281336739252\\
6	6.61909128715567\\
7	6.73571541162675\\
8	6.88817718490498\\
};
\addplot [color=mycolor2, line width=1.5pt, mark=square, mark options={solid, mycolor2}, forget plot]
  table[row sep=crcr]{%
1	1.19803276933086\\
2	2.1915890072837\\
3	2.89219540736027\\
4	3.29381637920633\\
5	3.48255188535249\\
6	3.57657211739086\\
7	3.64548002278835\\
8	3.72832303334331\\
};
\addplot [color=mycolor3, line width=1.5pt, mark=diamond, mark options={solid, mycolor3}, forget plot]
  table[row sep=crcr]{%
1	2.79640044106915\\
2	4.64730041958852\\
3	5.93405626437347\\
4	6.70227882383648\\
5	7.12907775364632\\
6	7.3883193112044\\
7	7.44031269063173\\
8	7.50055091430521\\
};
\addplot [color=mycolor4, line width=1.5pt, mark=triangle, mark options={solid, mycolor4}, forget plot]
  table[row sep=crcr]{%
1	2.2402690109093\\
2	3.57749419187283\\
3	4.39175226265294\\
4	4.81317574773141\\
5	4.98696680870351\\
6	5.04914570394654\\
7	5.07244290794775\\
8	5.09121249182443\\
};
\end{axis}

\begin{axis}[%
width=1.247\linewidth,
height=0.474\linewidth,
at={(-0.162\linewidth,-0.058\linewidth)},
scale only axis,
xmin=0,
xmax=1,
ymin=0,
ymax=1,
axis line style={draw=none},
ticks=none,
]
\end{axis}

\end{tikzpicture}%
}
    \vspace{-0.5cm}
    \caption{\ac{EE} and achieved \ac{SE} versus \ac{DAC} resolution across architectures.} 
    \label{fig:bits}
\end{figure}

Figure~\ref{fig:bits} isolates the role of \ac{DAC} resolution. \ac{MiLAC} places only $K$ stream-domain \acp{DAC} rather than $N$ antenna-domain ones, so its \ac{DAC} power scales as $K P_{\mathrm{DAC}}(b)$ versus $N P_{\mathrm{DAC}}(b)$ for digital. This $K/N$ ratio makes raising $b_{\mathrm{DAC}}$ substantially cheaper and yields a broader \ac{EE} peak that is less sensitive to bit depth, whereas digital's \ac{EE} drops more sharply at high $b_{\mathrm{DAC}}$. Hybrid-\ac{FC} and hybrid-\ac{SC} share the same $N_{\mathrm{RF}}$-scaled \ac{DAC} number and differ only in \ac{PS} power, so hybrid-\ac{SC} consistently outperforms hybrid-\ac{FC}. All four architectures use the same per-\ac{DAC} power model~\eqref{eq:pdac}, so their \ac{DAC} powers differ only through the number of \acp{DAC}, $K$ here versus $N_{\mathrm{RF}}$ and $N$. A different \ac{DAC} technology would change the constants in~\eqref{eq:pdac} for all architectures alike, so the conclusions of this sweep reflect the \ac{DAC} numbers rather than the specific \ac{DAC} technology.

\begin{figure}[t]
\centering
\resizebox{0.99\linewidth}{!}{
%
\definecolor{mycolor1}{rgb}{0.06600,0.44300,0.74500}%
\definecolor{mycolor2}{rgb}{0.86600,0.32900,0.00000}%
\definecolor{mycolor3}{rgb}{0.92900,0.69400,0.12500}%
\definecolor{mycolor4}{rgb}{0.12941,0.12941,0.12941}%

\begin{tikzpicture}

\begin{axis}[%
width=1\linewidth,
height=0.38\linewidth,
at={(0\linewidth,0\linewidth)},
scale only axis,
xmode=log,
xmin=0.25,
xmax=128.0,
xtick={0.25,0.5,1,2,4,8,16,32,64,128},
xticklabels={{0.25},{0.5},{1},{2},{4},{8},{16},{32},{64},{128}},
xminorticks=true,
xlabel style={font=\footnotesize\color{mycolor4}},
xlabel={MiLAC static-power scale factor},
ymin=0,
ymax=1400,
ylabel style={font=\footnotesize\color{mycolor4}},
ylabel={Energy efficiency (Mbit/J)},
axis background/.style={fill=white},
xmajorgrids,
xminorgrids,
ymajorgrids,
legend style={font=\footnotesize, legend cell align=left, align=left},
tick label style={font=\footnotesize}
]
\addplot [color=mycolor1, line width=1.5pt, mark=o, mark options={solid, mycolor1}]
  table[row sep=crcr]{%
0.25	223.862268169325\\
0.5	223.90294891333\\
1	223.916235617923\\
2	223.923264987022\\
4	223.925522273813\\
8	223.92762205609\\
16	223.929583760923\\
32	223.931422484828\\
64	223.933150525378\\
128	223.934778213375\\
};
\addlegendentry{Digital}

\addplot [color=mycolor2, line width=1.5pt, mark=square, mark options={solid, mycolor2}]
  table[row sep=crcr]{%
0.25	1208.08310110783\\
0.5	1133.71883634906\\
1	1013.4608595448\\
2	844.553125168912\\
4	646.316226510682\\
8	454.373098269948\\
16	297.770013204796\\
32	184.809779326988\\
64	110.375900974129\\
128	63.9423710484184\\
};
\addlegendentry{MiLAC-aided}

\addplot [color=mycolor3, line width=1.5pt, mark=diamond, mark options={solid, mycolor3}]
  table[row sep=crcr]{%
0.25	105.280714083124\\
0.5	105.844094195878\\
1	105.846382073228\\
2	105.846382073228\\
4	105.846382073228\\
8	105.846382073228\\
16	105.846382073228\\
32	105.846382073228\\
64	105.846382073228\\
128	105.846382073228\\
};
\addlegendentry{Hybrid-FC}

\addplot [color=mycolor4, line width=1.5pt, mark=triangle, mark options={solid, mycolor4}]
  table[row sep=crcr]{%
0.25	250.307822348085\\
0.5	256.469109459241\\
1	256.470760865448\\
2	256.470767538993\\
4	256.470767588948\\
8	256.470767589795\\
16	256.470767589815\\
32	256.470767589815\\
64	256.470767589815\\
128	256.470767589815\\
};
\addlegendentry{Hybrid-SC}

\end{axis}

\begin{axis}[%
width=1.245\linewidth,
height=0.473\linewidth,
at={(-0.162\linewidth,-0.058\linewidth)},
scale only axis,
xmin=0,
xmax=1,
ymin=0,
ymax=1,
axis line style={draw=none},
ticks=none,
]
\end{axis}

\end{tikzpicture}%
}
\vspace{-0.5cm}
\caption{\ac{EE} versus the scaled per-admittance drive-circuit power $P_{\mathrm{drv}}$.
}
\label{fig:static_sens}
\end{figure}

\begin{figure*}[t]
    \centering
    \subfloat[\ac{SE}-\ac{EE} tradeoff boundaries at three power budgets (left to right: $P_{\max}=20$~dBm, $30$~dBm, and $40$~dBm).]{%
        \resizebox{\linewidth}{!}{\input{figs/deepmimo_paper_frontier_overview.tikz}}%
        \label{fig:frontier_overview}%
    }

    \vspace{-0.8cm}

    \subfloat[\ac{SE}-\ac{EE} tradeoff boundaries at three antenna numbers (left to right: $N=16$, $64$, and $256$).]{%
        \resizebox{\linewidth}{!}{\input{figs/deepmimo_paper_frontier_by_N.tikz}}%
        \label{fig:frontier_N}%
    }
    \caption{\ac{SE}-\ac{EE} tradeoff boundaries across architectures. Each architecture is optimized with its own endpoint normalization via the weighted-sum method, but all curves are plotted in the raw $(\mathrm{SE},\mathrm{EE})$ plane.}
    \label{fig:frontier_combined}
\end{figure*}

Figure~\ref{fig:static_sens} examines robustness by scaling $P_{\mathrm{drv}}$ over two orders of magnitude. \ac{MiLAC}-aided beamforming keeps the highest \ac{EE} up to $16$ times the adopted value, that is up to $0.56$~mW per admittance, and falls below digital and hybrid-\ac{SC} only between $16\times$ and $32\times$, leaving a substantial margin for the architectural conclusion even if the true cost significantly exceeds the adopted reference.

\subsection{\ac{SE}-\ac{EE} Tradeoff Operating Range}\label{sec:tradeoff_discussion}

Having established \ac{MiLAC}-aided beamforming's \ac{EE} advantage, we now examine the \ac{SE}-\ac{EE} tradeoff boundaries to understand \emph{how much \ac{SE} is sacrificed for \ac{EE} improvement} and quantify the \emph{operating flexibility} offered to the system designer. We note that, on the same channel realization, the \ac{EE} endpoint of the tradeoff boundary upper-bounds the \ac{EE}-oriented design value of Figs.~\ref{fig:power}--\ref{fig:static_sens}, because the \ac{SCA}-based solver with endpoint refinement tightens local stationarity relative to the direct Dinkelbach-\ac{WMMSE} solver used in the \ac{EE}-only comparisons.

Figure~\ref{fig:frontier_combined}(a) plots the \ac{SE}-\ac{EE} boundaries at $20$, $30$, and $40$~dBm obtained via Algorithm~\ref{algo:tradeoff}. At $20$~dBm, only \ac{MiLAC}-aided beamforming displays a non-trivial boundary ($\Delta\mathrm{SE}\approx1.11$~bit/s/Hz, $\Delta\mathrm{EE}\approx147$~Mbit/J), as the other three architectures' circuit-power floors leave little room for rate-efficiency trading. At $30$~dBm all four develop meaningful frontiers, with \ac{MiLAC} spanning the broadest ($\Delta\mathrm{SE}\approx4.62$~bit/s/Hz, $\Delta\mathrm{EE}\approx683$~Mbit/J). At $40$~dBm this widens further, with \ac{MiLAC}'s $\Delta\mathrm{EE}$ reaching $870$~Mbit/J. The right endpoints of these boundaries also show how much of the \ac{SE} gap observed in Section~\ref{sec:ee_performance} is structural and how much is only the chosen operating point. At $P_{\max}=30$~dBm the \ac{EE}-optimal designs in Fig.~\ref{fig:power} are about $3.3$~bit/s/Hz apart, while the maximum-\ac{SE} endpoints, where the power backoff is absent, differ by only about $1.7$~bit/s/Hz ($10.1$ versus $8.4$). The structural cost of $\|\Ym\|_2\le 1$ is therefore only $1.7$ of the $3.3$~bit/s/Hz. The remaining $1.6$~bit/s/Hz reflects the reduced transmit power selected by the \ac{EE} objective and can be traded back for \ac{SE} by moving along the boundary.

\looseness=-1 Figure~\ref{fig:frontier_combined}(b) shows the boundary evolution with array size. At $N{=}16$ all four architectures display non-trivial frontiers, with \ac{MiLAC}-aided beamforming the widest ($\Delta\mathrm{EE}\approx441$~Mbit/J). By $N{=}64$, \ac{MiLAC}'s boundary lies clearly above all benchmarks ($\Delta\mathrm{EE}\approx710$~Mbit/J), while hybrid-\ac{SC} retains a moderate frontier ($96$~Mbit/J) and digital ($41$~Mbit/J) and hybrid-\ac{FC} ($9$~Mbit/J) compress. At $N{=}256$, \ac{MiLAC} still spans by far the widest frontier ($\Delta\mathrm{EE}\approx195$~Mbit/J), while digital and hybrid-\ac{FC} reduce to nearly single points as their $N$-scaled static power ($10$ and $22$~W, respectively) dominates the total power consumption, and hybrid-\ac{SC} retains only a narrow frontier ($\Delta\mathrm{EE}\approx9$~Mbit/J). This confirms that hybrid-\ac{SC}'s $N$-linear \ac{PS} scaling provides structural resilience over hybrid-\ac{FC} but still falls short of the passive-network approach, and that \ac{MiLAC}'s operating flexibility, not just the \ac{EE}-optimal point, grows with array size.

Together, Figures~\ref{fig:frontier_combined}(a) and~(b) show that \emph{\ac{MiLAC}-aided beamforming's operating flexibility expands along both the power and array-size axes}, so that an operator can move along the boundary to meet a quality-of-service \ac{SE} target while retaining the residual \ac{EE} gain.

\section{Conclusion}\label{sec:conclusion}

This paper established the first quantization-aware \ac{EE} optimization framework for \ac{MiLAC}-aided MU-MISO beamforming. Exploiting that \ac{MiLAC}'s stream-domain \ac{DAC} placement induces a full-covariance \ac{AQNM} distortion, we proved that the beamformer search can be restricted to the row space of the channel matrix without loss of optimality, yielding a $K{\times}K$ reduced-dimension parameterization that cuts the dominant per-iteration cost from $N$- to $K$-scaling without approximation. On this reduced problem, a Dinkelbach-\ac{WMMSE} algorithm with closed-form weighted power updates and a \ac{PGD}-based $\Ym$-update converges monotonically to a stationary point. For the \ac{SE}-\ac{EE} tradeoff, an alternative $K{\times}K$ variable removes the bilinearity and an auxiliary-variable lift solved by \ac{SCA} yields a convex per-iteration subproblem, with every limit point of the iterates satisfying the \ac{KKT} conditions at each weight, enabling the tradeoff boundary to be traced via the weighted-sum method.

On a site-specific DeepMIMO deployment, \ac{MiLAC} achieves the highest \ac{EE} across a $10$--$40$~dBm power sweep at a moderate \ac{SE} cost, outperforming both \ac{FC} and \ac{SC} hybrid benchmarks, and it keeps the highest \ac{EE} over digital and hybrid-\ac{SC} until the per-admittance drive-circuit power exceeds $0.56$~mW, $16$ times the adopted value. The tradeoff boundary analysis further shows that \ac{MiLAC}-aided beamforming expands the achievable operating region along both the power and array-size axes, an advantage digital and hybrid benchmarks cannot match.

Hardware-calibrated power measurements on \ac{MiLAC} prototypes are the most pressing next step. Further extensions include using the hybrid digital-\ac{MiLAC} structure to raise the \ac{SE} while preserving the \ac{EE} advantage, extending the feasible beamformer set and the power model to lossy, mutually coupled, or non-reciprocal admittance networks~\cite{Nerini2026Physics,Peng2026LossyBDRIS}, and evaluating multi-cell and multi-carrier deployments. Replacing the \ac{FC} network by a stem-connected one is a further direction: it needs an order of magnitude fewer admittances with no significant sum-rate loss~\cite{Nerini2025Reduced,Zhang2026StemBF}, but the sparser network also has fewer spare degrees of freedom to compensate for coarse admittance resolution and may need more drive power per admittance, so its net \ac{EE} effect deserves a separate study.

\appendices

\bibliographystyle{IEEEtran}
\bibliography{references}

\end{document}